\documentclass[12pt]{article}
\usepackage{amssymb, amsmath, amsthm}
\usepackage{lscape} 
\usepackage{mathdots}
\usepackage{thmtools}
\usepackage{mathtools}
\usepackage[makeroom]{cancel}
\usepackage{appendix}
\usepackage{bm}
\usepackage{bbm}
\usepackage{caption}
\usepackage{enumerate} 
\usepackage{graphicx}
\usepackage{natbib}
\usepackage{setspace}
\usepackage[a-1b]{pdfx}
\usepackage{hyperref} 
\usepackage{xcolor}
\hypersetup{
    colorlinks,
    linkcolor={red!50!black},
    citecolor={blue!50!black},
    urlcolor={blue!80!black}
}

\usepackage{tocloft}
\usepackage[explicit]{titlesec}
\usepackage{tikz}
\usetikzlibrary{bending}
\usetikzlibrary{shapes,snakes}
\usepackage{mathrsfs}
\usepackage{subcaption}
\usepackage[top=1.25in, bottom=1.25in, left=1.25in, right=1.25in]{geometry}

\pdfoutput=1

\newtheorem{lemma}{Lemma}
\newtheorem{theorem}{Theorem}

\newtheorem{proposition}{Proposition}
\newtheorem{definition}{Definition}
\newtheorem{assumption}{Assumption}

\def\ta{\theta}

\def\la{\lambda}

\def\phi{\varphi}

 \def\argmin{\mbox{argmin}}

\def\ul{\underline}
\def\bar{\overline}

\newcommand{\df}[1]{\textit{#1}}

\begin{document}
\title{Rationally Inattentive Statistical Discrimination: Arrow Meets Phelps\thanks{We thank the editor, four anonymous referees, Ilya Segal, and the audience at various conferences and seminars for their comments and suggestions.}}
\author{
Federico Echenique\footnote{Department of Economics, University of California, Berkeley, \href{mailto:fede@econ.berkeley.edu}{fede@econ.berkeley.edu}. }
\and Anqi Li\footnote{Department of Economics,  University of Waterloo, \href{mailto:angellianqi@gmail.com}{angellianqi@gmail.com}.}
}
\date{}

\maketitle
\thispagestyle{empty}
\begin{abstract}
When information acquisition is costly but flexible, a principal may rationally acquire information that favors one group over another. The former group faces incentives to invest in becoming productive, while the latter is discouraged from such investments.  The principal, in turn, ignores the productivity difference between groups unless the underinvested group surprises him with a genuinely outstanding outcome. We give conditions under which the discriminatory equilibrium is most preferred by the principal, despite all groups being ex-ante identical. Our results inform the discussion of affirmative action, implicit bias, and occupational segregation and stereotypes.

\vspace{.8cm}

\vspace{.3cm}

\end{abstract}

\newpage
\setcounter{page}{1}

\section{Introduction}\label{sec_intro}
We provide a new account of statistical discrimination. A demographic group is discriminated against in the labor market because its members rationally choose to underinvest in the skills needed to succeed. Their investment choice is reinforced by the endogenous allocation of an employer's  limited attention across groups, based on which workers estimate the returns to investing, and labor market decisions are made.  In equilibrium, discriminatory attention allocation and differing investment choices between ex-ante identical groups are mutually reinforcing.  Under some conditions, discriminatory equilibria are the most profitable to the employer.

The theory of statistical discrimination posits that groups with certain demographic traits are discriminated against in the labor market, because rational employers correctly infer that these groups should be treated differently. As an explanation for discrimination, the theory does not rely on bias or animosity,  but rather on the mechanism by which employers form discriminatory beliefs.

Economists have proposed two canonical models of statistical discrimination: the Arrovian model of coordination failure, and the Phelpsian model of information heterogeneity. \cite{arrow71,arrowJEP} argues that discrimination may arise as the result of coordination failure. One demographic collective, call it Group 1, expects to be discriminated against, and therefore does not undertake the costly investments that are needed to succeed in the labor market. Group 2 expects to be
favored and therefore invests. Employers, in turn, rationally favor Group 2 over Group 1 because the former is expected to invest and the latter is not. Such a discriminatory equilibrium is, typically, Pareto dominated by an impartial equilibrium where employers hold uniformly positive beliefs about all groups, and the latter all invest.

The second canonical model follows \cite{phelps72} (see also
\citealt{aigner1977statistical}) to argue that statistical discrimination emerges from
differing qualities of information. Groups 1 and 2 have the same, exogenous, skill distribution, but 
employers have access to better-quality information about members of Group 2
than of Group 1. As a result, members of Group 2 enjoy, on average, a
favorable treatment in the labor market.  Beyond a few prominent factors, such as shared language, culture backgrounds and social connections \citep{cornell1996culture}, the reasons behind the informational heterogeneity are often left unspecified. 

We combine the Arrovian and Phelpsian ideas, with the chief aim of endogenizing employers' information acquisition about their employees. In our model,  workers choose whether to undertake a costly investment that increases their likelihood of being productive. An employer selects a worker to promote,  based on his endogenously gleaned information about workers' productivity. Following the literature on rational inattention \citep{sims2003implications}, we model information acquisition as a costly signal structure that prescribes promotion recommendations to the employer.  In equilibrium, workers' incentives to invest are influenced by their likelihood of getting promoted, and their investment decisions in turn shape the signal structure that the employer uses to assess their productivity.  

We first show that there always exists an impartial equilibrium: analogous to the Pareto-dominant equilibrium in Arrow's model, but with the new feature that the information structure endogenously chosen by the employer is also impartial about groups. In an impartial equilibrium, there is neither Arrovian coordination failure nor Phelpsian information heterogeneity.

Our main results describe a discriminatory equilibrium, in which some groups choose not to invest because they are not expected to, while others do invest and correctly expect to be rewarded.  Crucially, these differing investment
decisions are mirrored in the employer's choice of a discriminatory signal 
structure --- one that  favors the group who invests, unless the
underinvested group is strictly more productive than the former; {\em Arrow meets Phelps.} In this way, the employer can efficiently deploy his limited attentional resources according to workers' investment decisions, focusing mainly on whether the underinvested group surprises him with a genuinely outstanding outcome. This differential treatment reinforces workers' expectations and a vicious circle is closed.

The following diagram plots the model's behavior against an attention cost parameter that captures how costly it is for the employer to acquire information:
\begin{center}
\begin{tikzpicture}

\draw[->] (0,0) -- (8,0) node[below,font=\scriptsize] {Attention cost};

\coordinate (ull) at (2,0);
\coordinate (lastar) at (4,0);
\coordinate (barl) at (6,0);

\draw[line width=2pt] (ull) -- (2,-0.04);
\draw[line width=2pt] (lastar) -- (4,0.04);
\draw[line width=2pt] (barl) -- (6,-0.04);

\draw[fill=red] (4,-0.04) rectangle (6,0.04);

\draw [thick, decorate,decoration={brace,amplitude=5pt}] (0,0.08)--
(3.95,0.08)
node [black,midway,above=6pt, font=\scriptsize,text width=3.5cm] {Unique
  impartial eq., high worker investment};

\draw [thick, decorate,decoration={brace,amplitude=5pt}] (4.05,0.09)--
(8,0.09) node [black,midway,above=6pt, font=\scriptsize,text width=3cm] {Unique
  impartial eq., low worker investment};

\draw [thick, decorate,decoration={brace,amplitude=5pt}] (6,-0.07) --
 (2,-0.07)
node [black,midway,below=6pt, font=\scriptsize,text width=3cm]
{Discriminatory eq.};

\draw[thin,->] (5,-0.05) -- (6,-1) node[right,font=\scriptsize,text width=3cm,red] {Employer prefers
  discriminatory eq.};
\end{tikzpicture}
  \vspace{-5pt}
\end{center}
Observe that an impartial equilibrium always exists and is generically unique. It features high worker investments when the attention cost parameter is low, and low investments when the attention cost is high. A discriminatory equilibrium emerges when the attention cost parameter is intermediate, and it is the most profitable equilibrium to the employer when it coexists with an impartial equilibrium that induces low worker investments. 

\label{paragraph_underyingmechansim}Our main results convey three basic messages. First, a discriminatory information structure can emerge endogenously, and be mirrored in workers' differential investment decisions, even though workers are ex-ante identical. The mechanism generating discrimination relies on a complementarity between attention asymmetry and investment asymmetry. 

Second, a discriminatory equilibrium may be strictly preferred by the employer to an impartial equilibrium. When attention is costly, impartiality implies uniformly low incentives to all workers. The employer prefers a discriminatory world in which one group is properly incentivized, while the other is only recognized when they are strictly more productive than the first. This allows the employer to be rationally inattentive and to save on attention cost, in addition to boosting revenue. Such welfare comparisons stand in contrast to Arrow's explanation of statistical discrimination  as a Pareto-dominated, ``bad,'' equilibrium to all parties involved. 

Third, the degree of discrimination in the most profitable equilibrium is nonmonotonic in the attention cost parameter. Our model has no discriminatory equilibrium when the attention cost parameter is (close to) zero or very high. Discrimination arises, and constitutes the most profitable equilibrium to the employer, when the attention cost parameter is intermediate. Consequently, the effects of lowering the attention cost parameter on the equilibrium degree of discrimination and players' welfare are in general ambiguous. Our comparative statics result speaks to the de-biasing programs used by real-world organizations to address discrimination. These programs are based on the conventional wisdom in social psychology that limited attention triggers implicit biases \citep{greenwald1995implicit, macrae2000social}, but their success has been mixed \citep{greenwald2020implicit}, which seems consistent with out findings. See Section~\ref{sec:implications} for a detailed discussion.

Our model not only adds to the theory of statistical discrimination; it also provides a tractable framework to discuss various policy issues, as well as phenomena associated with labor market discrimination. In Section~\ref{sec_aa}, we use our model to evaluate the effectiveness of affirmative action quotas in addressing discrimination. 
We show that mandating a quota that requires members of different groups be promoted  with equal probability eliminates discriminatory equilibria, without impacting on impartial equilibria. Furthermore, it does not generate new, patronizing equilibria as a byproduct --- a contrast with the previous literature \citep{coateloury1993}.
A second application to occupational discrimination can be found in the Online Appendix, where the employer assigns workers to perform multiple, ex-ante symmetric tasks using stereotypical screening. This is mirrored in workers' differential investments in task-specific skills,  which, in the most profitable equilibrium to the employer, gives rise to occupational segregation and stereotypes.

\section{Model}\label{sec_model}
We study a game between three players: a principal, and two agents who are called Michael
($m$) and Wendy ($w$). The principal must choose one of the agents to promote. The promotion decision serves to induce the agents to exert effort so as to be more productive. It delivers a unit benefit to the chosen agent, as well as the agent's productivity to the principal. One can broadly interpret the promotion opportunity as a reward (e.g., salary raise, employee recognition,  favorable task assignment) that motivates agents to undertake costly investments.  For the sake of concreteness, we shall stick to the interpretation of promotion throughout.

Specifically, each agent $i\in\{m,w \}$ chooses a level of \df{effort} $\mu_i\in
\{\ul\mu,\bar\mu \}$, with $0<\ul\mu<\bar\mu<1$, at a cost $C(\mu_i)$. Suppose that $C(\ul\mu)=0$ and that $C(\bar\mu)=C \in (0,1/2)$. The effort $\mu_i$ generates a random \emph{productivity} $\tilde{\ta}_i$ for agent $i$, with $\mu_i$ being the probability that $\tilde{\ta}_i=1$ and $1-\mu_i$ the probability that $\tilde{\ta}_i=0$. Given the profile $\bm{\mu}\coloneqq (\mu_m,\mu_w)$, productivities are drawn independently across the agents. 

The principal does not observe the realizations of $\tilde{\ta}_m$ and $\tilde{\ta}_w$, but can acquire information about them. Information, however, is costly. Given the information that the principal gleans about $\tilde{\bm{\ta}}\coloneqq (\tilde{\ta}_m,\tilde{\ta}_w)$, he chooses whom to promote. Specifically, the principal selects $a\in\{0,1\}$, where $a=0$ means that Wendy is promoted, and $a=1$ means that Michael is promoted.

Information acquisition is modeled as the choice of a signal structure $\pi:\{0,1\}^2\to\Delta(S)$, which maps each profile of productivity values to a random signal taking values in a set $S$. We assume that $S$ is finite
and that $|S| \geq 2$; later we shall demonstrate that these assumptions about $S$ are without loss of generality. Otherwise we impose no restriction on the signal structure, in order to model attentional flexibility and to study its impact on statistical discrimination (as suggested by the supporting evidence reviewed in Section \ref{sec_literature}). A promotion rule is a function $a:S\to \Delta(\{0,1\})$, which maps each signal realization to a (random) decision on whether to promote Michael or Wendy.  The profile $(\pi,  a(\cdot))$ of signal structure and promotion rule fully captures the principal's strategy. 

Given a profile $\bm{\mu}$ of effort choices by the agents, the principal's expected payoff is \[ 
\mathbb{E} \left[\tilde a\tilde \ta_m + (1-\tilde a)\tilde \ta_w \mid  \bm{\mu},\pi,a(\cdot) \right] - \la  I(\pi\mid \bm{\mu}),
\] where $\la>0$ parameterizes the cost of information acquisition, and is hereinafter referred to as the \emph{attention cost parameter};  $I$ is the mutual information (or reduction in Shannon entropy) between the random productivity profile $\tilde{\bm{\ta}}$ and the random signal generated by $\pi$. In words, the principal's payoff equals the productivity of the promoted agent,
which is estimated according to the information generated by the signal structure of his choice. As the latter becomes more informative of agents' productivity, the cost of information acquisition increases. 

The game begins with the principal and agents moving simultaneously: the former
chooses a signal structure and a promotion rule, and the latter make effort choices. After agents have made their choices, productivity and signals are realized. Then the principal's promotion decision is implemented. When choosing an agent to promote, the principal observes neither agents' efforts, or productivity, thus facing a moral hazard problem.  Agents do not observe the principal's choice of the signal structure or promotion rule --- an assumption that reflects the subjective nature of employee evaluation and promotion in practice. A variation of the game sequence, with the principal first committing to a signal structure, is explored in Online Appendix O.2.

We examine Bayes Nash equilibria in which agents adopt pure strategies (hereinafter, \emph{equilibrium} for short). When multiple equilibria coexist, we characterize them all, with a particular focus on the \emph{most profitable equilibrium to the principal}. Our equilibrium selection mechanism is standard in the contract theory literature and best captures situations in which the principal holds the bargaining power to influence the selection of equilibrium. See, however, Section \ref{sec:implications} for agents' preferences over multiple equilibria, and Online Appendix O.4 for equilibria in mixed strategies.

\section{Results}\label{sec_overview}
To proceed with our main results, we first present some preliminary concepts, followed by statements of the results and their intuitions, and finally their formal analysis. 
\vspace{-10pt}
\subsection{Preliminaries} We first simplify the principal's strategy in a manner that is now standard in the rational inattention literature; a seminal reference is \cite{matvejka2015rational}.  Define $\Delta \theta
\coloneqq \theta_m-\theta_w$ as the differential productivity value between $m$ and $w$, and note that $\Delta \theta \in \{-1,0,1\}$.  For any
given effort profile $\bm{\mu}$, rewrite the principal's expected payoff  as 
\begin{equation}\label{eqn_mainproblem}
\underbrace{\mathbb{E} \left[\tilde{a} \Delta \tilde{\theta} \mid  \bm{\mu},\pi,a(\cdot) \right]+\mu_w}_\text{Expected revenue} - \la I(\pi \mid \bm{\mu}), 
\end{equation}
where $\mu_w$ is $w$'s expected productivity,  and $\tilde{a} \Delta\tilde{\theta}$ is the (random) change in the principal's revenue by promoting $m$ rather than $w$. Crucially, the expected revenue depends on the principal's strategy $(\pi, a(\cdot))$ only through $\Delta \tilde{\theta}$.  Therefore,  we may restrict attention to signal structures that prescribe a (random) \emph{promotion recommendation} to the principal based on the differential productivity value between $m$ and $w$, i.e., $\pi: \{-1,0,1\} \rightarrow \Delta(\{m,w\})$, as any information beyond the differential productivity is redundant, and therefore shouldn't be acquired.  Furthermore, any optimal signal structure, if nondegenerate, must prescribe promotion recommendations that the principal strictly prefers to obey, i.e., $a(m)=1$ and $a(w)=0$.\footnote{The reason is that we can always label promotion recommendations in such a way that the principal weakly prefers to obey them. When an optimal signal structure is nondegenerate but violates strict obedience, the principal must have a (weakly) preferred candidate regardless of the promotion recommendations he receives. Consequently, he can always promote that agent without acquiring information to save on attention cost, which contradicts the assumption that the optimal signal structure is nondegenerate. \label{fn_strictobedience}} Consequently, we can represent the principal's strategy by $\pi: \{-1,0,1\} \rightarrow [0,1]$, where each $\pi(\Delta\theta)$, $\Delta \theta \in \{-1,0,1\}$, specifies the probability that $m$ is recommended for promotion when the differential productivity value between $m$ and $w$ equals $\Delta \theta$.

Using $\pi$ we can talk formally about the presence or absence of discrimination.

\begin{definition}
A signal structure $\pi$ is \emph{impartial} if the probability of promoting an agent
depends \emph{only} on his or her productivity difference with the other agent, and
\emph{not} on agents' identities. That is, $\pi(\Delta \theta)=1-\pi(-\Delta \theta)$ $\forall \Delta \theta \in \{-1,0,1\}$. Otherwise $\pi$ is  \emph{discriminatory}.
\end{definition} 

\begin{definition}
An equilibrium is \emph{impartial} (resp.\ \emph{discriminatory}) if the equilibrium signal structure is impartial (resp.\ discriminatory).
\end{definition}

We will show that an impartial equilibrium must induce the same level of effort from both agents, while a discriminatory equilibrium must induce different levels of effort from the two agents. By symmetry, it is without loss of generality to focus on discriminatory equilibria that induce high effort from $m$ and low effort from $w$ --- a convention we will follow in the remainder of the paper.

Lastly we introduce a regularity condition.

\begin{assumption}\label{assm_regularity}
(i) $\bar\mu+\ul\mu>1$ and (ii) $\frac{C}{\bar\mu-\ul\mu}<\frac{\bar\mu(1-\bar\mu)}{\bar\mu+\ul\mu-2\bar\mu \ul\mu}$.
\end{assumption}
Assumption \ref{assm_regularity} stipulates that when an agent exerts different levels of effort, the probabilities of experiencing a high productivity shock must be sufficiently high, and the associated cost of exerting high effort must be sufficiently low. The roles played by Assumption  \ref{assm_regularity} in our analysis are discussed in Section~\ref{sec:implications}.\footnote{Part (ii) of Assumption \ref{assm_regularity} is stronger than $C<1/2$, a condition that is needed for sustaining the high effort profile in any equilibrium and is maintained throughout the paper to make the analysis interesting. When information acquisition is costless, an agent earns an expected payoff of $1/2-C$ under the high effort profile and can always secure a nonnegative payoff by exerting low effort.  Shirking is deterrable only if $C<1/2$.} 

\subsection{Main results}\label{sec:mainresults} 
Our main results are twofold. The first
concerns the existence and uniqueness of impartial and discriminatory equilibria.  The second pinpoints the most profitable equilibrium to the principal.

\begin{theorem}\label{thm_existenceuniqueness}
For any $C$, $\bar\mu$, and $\ul\mu$ that satisfy Assumption~\ref{assm_regularity}, there exist values $\ul\la, \bar\la$, and $\la^*$ of the
attention cost parameter such that $0<\ul\lambda<\bar\lambda<+\infty$ and
$0<\la^*<+\infty$, and the following statements are true:
\begin{enumerate}[(i)]
\item An impartial equilibrium always exists.  For all 
  $\lambda \neq \lambda^*$, the impartial equilibrium  is unique; it sustains the high effort profile $(\bar{\mu}, \bar{\mu})$ if the attention cost parameter is
  low,  i.e., $\lambda<\lambda^*$, and the low effort profile $(\ul{\mu}, \ul{\mu})$ if the attention cost parameter is high,  i.e., $\lambda>\lambda^*$. 
\item A discriminatory equilibrium exists if and only if the attention cost parameter
  is intermediate, i.e., $\lambda \in [\ul\lambda, \bar\lambda]$.  Whenever a discriminatory equilibrium exists,  there is a unique discriminatory equilibrium that sustains $(\bar\mu, \ul\mu)$.
\item $\ul\lambda<\lambda^*$ always holds.  $\lambda^*<\bar\lambda$ holds if, in addition, $\ul\mu>1/2$ and Condition (\ref{eqn_regularity}) in Appendix \ref{sec_proof} is satisfied.\footnote{The role of Condition (\ref{eqn_regularity}) is discussed in Section \ref{sec:implications}. Numerical analysis confirms it can hold simultaneously with Assumption \ref{assm_regularity}. }
\end{enumerate}
\end{theorem}

\begin{theorem}\label{thm_mostprofitable}
Let everything be as in Theorem  \ref{thm_existenceuniqueness}, and suppose that $\lambda^*<\bar\lambda$. Then the most profitable equilibrium to the principal is discriminatory if and only if $\lambda \in (\lambda^*, \bar\lambda]$. 
\end{theorem}

To develop the intuitions behind these results, we first restrict the
principal to using impartial signal structures. Under this restriction, the signal
acquired by the principal becomes less informative about agents' productivity, in
the sense of Blackwell, as the attention cost parameter increases. Agents best
respond by exerting high effort when the attention cost parameter is low, and low
effort when the attention cost parameter is high. The symmetry in agents' effort choices, in
turn,  justifies the use of an impartial signal structure to begin with. The two regimes are separated by the threshold value $\lambda^*>0$, at which the game has two impartial equilibria. For all $\lambda \neq \lambda^*$, the impartial equilibrium is unique.

We next allow the principal to use discriminatory signal structures. As an illustration, consider the
numerical example in Table~\ref{table1},  which takes a  discriminatory effort profile as given and solves for the optimal signal structure, i.e., one that maximizes the principal's expected profit.
\begin{table}[hht]
\centering
\caption{Optimal signal structure for 
$\bm\mu=(\bar\mu, \ul\mu)=(.8, .6)$, $\lambda=.3$.}
\begin{tabular}{ |c|c|c|c| }
\hline
$\Delta \theta$ & 1& 0 & -1 \\ \hline
$\mathbb{P}(\Delta \theta \mid \bm\mu)$ & .32 & .56 & .12 \\ \hline
$\pi(\Delta \theta)$ & .98 & .74 & .09 \\ \hline
\end{tabular}
\vspace{-5pt}
\label{table1}
\end{table} 
\noindent Since $m$ is known to work harder than $w$, promoting $m$ over $w$ is
the safer choice for the principal. In consequence,  a rationally inattentive principal will favor $m$
unless $w$ is strictly more productive.   While
$w$ is strongly favored by the principal when she is strictly more productive than
$m$ ($\pi(-1)=.09$), that event occurs with a small probability because $m$
works harder than $w$.  $w$ is treated unfavorably otherwise. In particular, and
importantly, this occurs when she is as productive as $m$ ($\pi(0)=.74$).  A benefit stemming from this distortion is that the principal doesn't need to carefully distinguish between whether $m$ is more productive than,
or equally productive as $w$ (indeed $\pi(1)=.98$ is not very different from
$\pi(0)=.74$) --- a practice that saves on attention cost. At the same time,
the signal structure still does a decent job in selecting the most
productive agent, as it generates an expected revenue of $.90$, compared to the expected
revenue $.92$ in the benchmark case where information acquisition is costless.  

Turning to agents' incentives to invest, under the above numerical assumptions,  $w$ can only increase her winning probability by 
\[\Delta\mu [\bar\mu(\pi(1)-\pi(0))+(1-\bar\mu)(\pi(0)-\pi(-1))]=.081\]
if she exerts high effort rather than low effort, holding everything else constant. The analogous decrease for $m$ is 
\[\Delta\mu [(1-\ul\mu)(\pi(1)-\pi(0))+\ul\mu(\pi(0)-\pi(-1))]=.098\]  
if he shirks rather than work.  If $C \in (.081, .098)$, then it is indeed optimal for  $m$ to exert high effort and
$w$ low effort. In turn, this justifies the principal's use of the discriminatory signal structure that favors $m$.

Taken together, our main results present an important lesson: Conducting discriminatory performance evaluations allows the principal
to be rationally inattentive and to sustain a discriminatory effort profile in
equilibrium when the attention cost parameter is intermediate, i.e., $\lambda \in [\ul\lambda, \bar\lambda]$. Compared to the impartial, low-effort equilibrium,  the discriminatory equilibrium enjoys a revenue advantage because it still induces one agent to work,  as well as a cost advantage because it is cheaper to implement.  
If, instead, the impartial, low-effort equilibrium is being played, then the choice of which agent to promote is a priori nonobvious because both agents work equally hard. Consequently, the principal must compare them carefully at a significant cost.  For these reasons, the discriminatory equilibrium is more profitable than the impartial, low-effort equilibrium whenever they coexist, i.e., $\lambda \in [\lambda^*, \bar\lambda]$.

The comparison between the discriminatory equilibrium and the impartial, high-effort equilibrium is more delicate. While the former may or may not have a cost advantage, it definitively suffers a revenue disadvantage compared to latter. It turns out that the revenue concern is always of a first-order importance, which makes the discriminatory equilibrium less profitable than the impartial, high-effort equilibrium whenever they coexist, i.e., $\lambda \in [\ul\lambda,\lambda^*]$.

\subsection{Analysis}\label{sec_analysis}
This section provides a more formal analysis of Theorems~\ref{thm_existenceuniqueness} and \ref{thm_mostprofitable}. We begin by characterizing players' best response functions, followed by a complete characterization of equilibria.  We next outline the key steps in comparing equilibrium profitability, leaving the technical details to Appendix \ref{sec_proof}. 
\vspace{-12pt}

\paragraph{Notations.} For an arbitrary signal structure $\pi$, write $\bar\pi$ for the probability that $m$ is recommended for promotion, as well as $X$ for $\pi(1)-\pi(0)$ and $Y$ for $\pi(0)-\pi(1)$. $X$ represents the increase in $m$'s promotion probability when he strictly outperforms $w$ rather than tying with her performance, and $Y$ represents the decrease in $m$'s promotion probability when he strictly underperforms $w$ rather than tying with her performance. $\pi$ is impartial if and only if $\pi(0)=1/2$  and $X=Y$.

It is easier to work with $\gamma \coloneqq \exp(1/\lambda)$ rather than the attention cost parameter $\lambda$, so we will follow this convention in what follows. $\gamma$ is strictly decreasing in $\lambda$,  $\gamma \rightarrow +\infty$ as $\lambda \rightarrow 0$, and $\gamma \rightarrow 1$ as $\lambda \rightarrow +\infty$. 

It is also convenient to express Part (ii) of Assumption \ref{assm_regularity} as $c<\bar\mu (1-\bar\mu)/(A+B)$, where $c \coloneqq C/\Delta \mu$ represents the \emph{effective cost} of exerting high effort, normalized by the change $\Delta\mu \coloneqq \bar\mu - \ul \mu$ in the probability of experiencing a high productivity shock when working rather than shirking. $A\coloneqq \bar\mu(1-\ul\mu)$ denotes the probability that $m$ strictly outperforms $w$, i.e., $\Delta \theta=1$, under the discriminatory effort profile $(\bar\mu, \ul\mu)$, and $B$ denotes the probability that he strictly underperforms $w$, i.e., $\Delta \theta=-1$, under the same effort profile.  Importantly, $A>B$. 

\vspace{-10pt} 
\paragraph{Best response functions.}\label{sec_br}
 Consider first  the problem faced by the principal,  holding agents' effort profile $\bm{\mu}$ fixed.  Call the solution to this problem the \emph{optimal signal structure for $\bm{\mu}$}. By \cite{matvejka2015rational}, this signal structure is either degenerate, satisfying $\pi(\Delta \theta) \equiv 0$ or $1$, or it is nondegenerate and satisfies  $\pi(\Delta \theta) \in (0,1)$ $\forall \Delta \theta$. The next lemma solves for the optimal signal structure for every effort profile.

\begin{lemma}\label{lem_optimalsignal}
\begin{enumerate}[(i)]
\item The optimal signal structure for $(\bar{\mu}, \bar{\mu})$ or $(\ul{\mu}, \ul{\mu})$ is nondegenerate and impartial.  It satisfies $\bar{\pi}=\pi(0)=1/2$ and $X=Y=g(\gamma)$, where
\[g(\gamma)
\coloneqq \frac{\gamma-1}{2(\gamma+1)} \text{ satisfies } g>0 \text{ and } \frac{dg(\gamma)}{d\lambda}<0 \text{ } \forall \lambda>0.\]
\item The optimal signal structure for $(\bar{\mu}, \ul{\mu})$ is degenerate if $\lambda \geq \breve{\lambda} \coloneqq (\ln (A/B))^{-1}>0$, and it is nondegenerate otherwise.  In the second case, the signal structure is discriminatory and satisfies $\bar{\pi}=\pi(0)=(\gamma A-B)[(\gamma-1)(A+B)]^{-1} \in (1/2,1)$ and $X=f(\gamma)<Y=Af(\gamma)/B$, where 
\[f(\gamma) \coloneqq \frac{(\gamma A-B)(\gamma B-A)}{(\gamma^2-1)(A+B)A} \text{ satisfies } f>0 \text{ and } \frac{df(\gamma)}{d\lambda}<0 \text{ } \forall \lambda \in (0, \breve{\lambda}). \]
\end{enumerate}
\end{lemma}

Lemma \ref{lem_optimalsignal} conveys four important messages. First, if an optimal signal structure is nondegenerate,  the conditional probability that it recommends $m$ for promotion is strictly increasing in the differential productivity between $m$ and $w$, i.e., $X, Y>0$.  

Second, when agents are equally productive, the conditional probability that an optimal signal structure promotes $m$ equals the average probability, i.e., $\pi(0)=\bar{\pi}$.

Third, the optimal signal structure is impartial when both agents exert the same level of effort, and it is discriminatory otherwise. Consequently, any impartial equilibrium must induce the same level of effort from the agents, while any discriminatory equilibrium must do the opposite. 

Finally, as the attention cost parameter $\lambda$ increases,  any optimal signal structure becomes ``noisier,'' in that the conditional probabilities that it recommends the most productive agent for promotion become more concentrated around the average probability, i.e., $X$ and $Y$ are both decreasing in $\lambda$.  For impartial optimal signal structures, we can strengthen the notion of informativeness to that of Blackwell. 

We next turn to agents' best response functions. The next lemma solves for an agent's best response to a given signal structure and the other agent's effort choice. 

\begin{lemma}\label{lem_optimaleffort}
Fix any signal structure $\pi$.  For any $\mu_w \in \{\ul\mu, \bar\mu\}$, $m$ prefers to exert high effort rather than low effort if and only if 
\[(1-\mu_w) X + \mu_w Y \geq c.\]
For any $\mu_m \in \{\ul\mu, \bar\mu\}$, $w$ prefers to exert high effort rather than low effort if and only if  
\[\mu_m X + (1-\mu_m) Y \geq c.\]
\end{lemma}

From $m$'s perspective, $X$ is a carrot that is effective when $w$ has a low productivity (hence $m$ can outperform $w$ and raise his chance of winning), and $-Y$ is a stick that is effective when $w$ has a high productivity.  The overall incentive power that a signal structure provides to him is $(1-\mu_w)X+\mu_w Y$. By exerting high effort rather than low effort, $m$ can increase his chance of getting promoted by $\Delta \mu [(1-\mu_w)X+\mu_w Y]$.  In the case where $(1-\mu_w)X+\mu_wY$ exceeds the effective cost $c \coloneqq C/\Delta \mu$ of exerting high effort, exerting high effort is optimal for $m$. 

The problem faced by $w$ can be solved analogously.  In case $\pi$ is an optimal signal structure, Lemma \ref{lem_optimalsignal} implies that sustaining high effort becomes harder as $\lambda$ increases. 

\vspace{-10pt}

\paragraph{Equilibrium construction.}\label{sec_eqm} Consider first the case of impartial equilibria, in which the optimal signal structure satisfies $X=Y=g(\gamma)$.  It induces both agents to exert high effort if $g(\gamma) \geq c$, and low effort if $g(\gamma)\leq c$.  The two regimes are separated by a single threshold: 
 \[\lambda^* \coloneqq (\ln g^{-1}(c))^{-1},\]
 at which the game has two impartial equilibria. For all $\lambda \neq \lambda^*$, the impartial equilibrium is unique.

The discriminatory case is illustrated by Figure~\ref{figure2}.  
In order to induce high effort from $m$ and low effort from $w$,  the profile $(X,Y)$ must lie above the blue line segment and below the black line segment.  
\begin{figure}[ht!]
\begin{center}
\scalebox{1}{
\begin{tikzpicture}[scale=.9]
\draw[->] (0,0)--(5.9,0);
\node[right, font=\scriptsize] at (5.9,0){$X$};
\draw[->] (0,0)--(0,5.9);
\node[above, font=\scriptsize] at (0,5.9){$Y$};

\draw[fill=black,opacity=0.1]  (2,2) -- (0,3.333) -- (0,5.714) -- cycle;
\node[right, font=\scriptsize] at (-.1, 3.6){$IC_{m} \wedge IC_{w}$};
 \draw[-,blue, thick] (0, 3.333)--(5,0);
 \draw[-,black,thick] (0, 5.714)--(3.07,0);
\draw[->] (.8, 5.2)--(.5,5);
\node[black,  font=\scriptsize] at (2.3,  5.2){$\overline{\mu} X +(1-\overline{\mu})Y=c$};
\draw[->] (3.4,1.45)--(3.1,1.35);
\node[right, font=\scriptsize] at (3.3,1.5){ $(1-\underline{\mu}) X +\underline{\mu} Y = c$};

\draw[-,red,thick] (0,0)--(2.3,4.5);
\node[right, font=\scriptsize] at (2.3,4.5){ $Y = AX/B$};

\draw[dashed] (1.49,0)--(1.49,2.96);
\node[below, font=\scriptsize] at (1.6,0){ $\overline{X}$};
\draw[dashed] (1.25,0)--(1.25,2.5);
\node[below, font=\scriptsize] at (1.20,0){$\underline{X}$};

\draw[dashed](0,0)--(2,2);
\node[right, font=\scriptsize] at (0.2,0.2){$45^{\circ}$};
\end{tikzpicture}
}
\caption{Equilibrium in the discriminatory case.}\label{figure2}
\end{center}
\vspace{-15pt}
\end{figure}
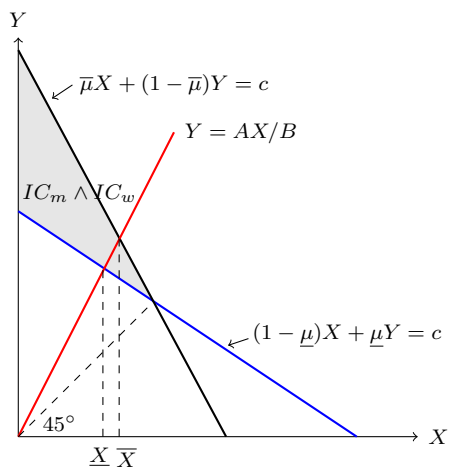
Under the assumption that $\bar\mu+\ul\mu>1$, the intersecting area, marked grey, lies above the 45-degree line, indicating a reliance on the carrot rather than the stick to provide incentives, i.e., $Y>X$. This configuration aligns with the optimal discriminatory signal structure, which features primarily a stick, i.e., $Y=AX/B>X$. Formally, the red ray on which the signal structure lies crosses the grey area twice,  at $(\ul X,  A\ul X/B)$ and $(\bar X, A \bar X/B)$, respectively. 
Thus for any $X=f(\gamma) \in [\underline{X}, \overline{X}]$, the profile $(X, AX/B)$ can arise in an equilibrium.  The last condition is equivalent to $\lambda \in [\ul\lambda, \bar\lambda]$, where 
 \[\ul\lambda \coloneqq (\ln f^{-1}(\bar X))^{-1} \text{ and }\bar\lambda \coloneqq (\ln f^{-1}(\ul X))^{-1}.\]
 
It remains to sign and rank $\ul\lambda$, $\lambda^*$, and $\bar\lambda$. This step is technical and is relegated to Appendix \ref{sec_proof}.  

\vspace{-10pt}
\paragraph{Equilibrium profit.} Clearly, the high effort profile generates more expected revenue to the principal than the discriminatory effort profile, followed by the low effort profile. The next lemma ranks the attention costs incurred by these effort profiles.

\begin{lemma}\label{lem_costrank}
Let everything be as in Theorem \ref{thm_mostprofitable}. Then the low effort profile incurs the highest attention cost, while the attention cost incurred by the discriminatory effort profile and that of the high effort profile cannot be generally ranked.  
\end{lemma}

Taken together, we conclude that the discriminatory equilibrium is the most profitable to the principal when $\lambda \in (\lambda^*, \bar\lambda]$, and that the impartial, high-effort equilibrium is more profitable than the impartial, low-effort equilibrium at $\lambda=\lambda^*$. In Appendix \ref{sec_proof}, we use more nuanced characterizations of the revenue and attention cost functions to show that the discriminatory equilibrium is less profitable than the high-effort, impartial equilibrium when $\lambda \in [\ul\lambda, \lambda^*]$. Combining the arguments yields Theorem \ref{thm_mostprofitable}.

\section{Implication and discussion}\label{sec:implications} 
We now examine the comparative statics and welfare consequences of our results, as well as their implications for labor market discrimination. This is followed by a discussion of the roles of key model assumptions in our analysis.  
\vspace{-10pt}
\paragraph{Gender and racial gap in subjective performance evaluation.} 
Gender and racial stereotypes continue to disadvantage women and minorities through biased subjective performance appraisals \citep{mackenzie2019most}.
Years of sociological research reveal that women get shorter, more vague, and less constructive critical feedback \citep{wynn2018},  and that they are held to higher performance standards and face increased scrutiny when being evaluated \citep{correll2016research}.  

Our model speaks to these stylized facts, predicting that minorities are rated more harshly than majorities in the discriminatory equilibrium, and that the only way for minorities to gain recognition from the employer is to surpass the majorities by a large margin.  Such a hurdle discourages minorities from undertaking costly investments, resulting in less frequent promotions and lower earnings on average.

\vspace{-10pt}
\paragraph{Implicit bias, stereotype, and the effectiveness of de-biasing programs.}\label{sec:debias}
Many scholars, across multiple disciplines, have
advanced the notion that limited attention triggers implicit biases and
stereotypes. The idea is that in attempting to make sense of the complex world surrounding us, we
regularly construct and use categorical representations to achieve efficient deployment of processing resources, causing implicit biases and stereotypes to emerge \citep{greenwald1995implicit, macrae2000social}. Evidence on the connection between attention and implicit discrimination abounds in areas such as labor economics, criminal justice, education, and healthcare \citep{bertrand2005implicit, eberhardt2020biased, warikoo2016examining, chapman2013physicians}. Notably, \cite{bertrand2005implicit} interpret the well-known study of discrimination through African-American names by \cite{bertrand2004emily} as evidence that time-constrained recruiters may allow implicit biases to guide their decisions.

Our model formalizes a causal link between limited attention and implicit discrimination. It predicts a nonmonotonic relation between the attention cost parameter and the equilibrium degree of discrimination; recall the statement of Theorem~\ref{thm_existenceuniqueness}, or the diagram in the  introduction. The nonmonotonic nature of the comparative-statics speaks to the varying effectiveness of the de-biasing training programs used by real-world organizations to address implicit discrimination. These programs share a common instruction: Every time a supervisor is supposed to make decisions that might adversely affect the supervisees (e.g., conduct performance evaluations),  it is reminded that he or she should ``slow down,  meditate, and follow elaborate procedures,'' so that decisions are made based on deliberations and factual information rather than quick instincts \citep{eberhardt2020biased}.\footnote{
The idea of using attention to intercept implicit discrimination has seen applications in other contexts.  Recently, the Oakland Policy Department adjusted its foot pursuit policy so that officers could no longer follow suspects as they run into backyards or blind alleys. Instead, officers were instructed to ``step back, slow down, call for backup, and think it through.'' Relatedly, Meta's ``Nextdoor Neighbor,'' a social network for residential neighbors to communicate through, recently started asking its users to provide detailed descriptions about the suspicious activities they wish to report to the system, because ``adding frictions allows users to act based on information rather than instinct'' \citep{eberhardt2020biased}. } The idea (and hope) behind is that one could alter the principal's (shadow) cost of acquiring information (as captured by $\lambda$),  through interventions such as increasing the amount of time committed to conducting performance evaluations.


By now, numerous corporations and nonprofit organizations have implemented programs of a similar sort, and tons of data are available for program evaluation.  Results of recent meta analysis are mixed,  leading \cite{greenwald2020implicit} to conclude that ``The popular media often suggests relying on one’s own mental resources to intercept implicit biases. Convincing evidence for the effectiveness of these strategies is not yet available in peer-reviewed publications.'' Our results put these mixed findings into perspective,  suggesting that rather than to abandoning the premise that limited attention triggers implicit biases, an alternative way to reconcile the aforementioned findings is to recognize that the exact relation between attention and implicit discrimination is more nuanced than previously thought. 

\vspace{-10pt}

\paragraph{Welfare.} An important consequence of our results is that one cannot Pareto rank the various kinds of equilibria.\footnote{Meanwhile, holding the equilibrium fixed, the principal's expected profit is decreasing in $\lambda$. Agents' welfare is invariant with $\lambda$ if the equilibrium is symmetric. If the equilibrium is discriminatory, then $m$'s expected payoff $\bar\pi-C$ is increasing in $\lambda$, while $w$'s expected payoff $1-\bar\pi$ is decreasing in $\lambda$. \label{fn_eqmwelfare}}  
\begin{proposition}\label{prop_welfare}
Let everything be as in Theorem \ref{thm_mostprofitable}. 
\begin{enumerate}[(i)]
\item On $[\lambda^*, \bar\lambda]$, the principal and $w$ prefer the impartial, high-effort equilibrium to the discriminatory equilibrium, while $m$'s preference is the opposite. 

\item  On $[\ul\lambda, \lambda^*]$, the principal and $m$ prefers the discriminatory equilibrium to the impartial, low-effort equilibrium, while $w$'s preference is the opposite. 

\end{enumerate}
\end{proposition}

The result concerning the principal is immediate from Theorem \ref{thm_mostprofitable}. As for agents, one can show that $m$ most prefers the outcome generated by the discriminatory equilibrium, followed by the impartial, low-effort equilibrium, and finally the impartial, high-effort equilibrium. $w$ most prefers the outcome generated by the impartial, low-effort equilibrium, followed by the impartial, high-effort equilibrium, and finally the discriminatory equilibrium. Depending on the exact welfare weights of the principal and agents, reduced discrimination may either enhance or undermine social welfare. This finding further complicates the picture painted by our results,  suggesting that the aforementioned de-biasing programs might not only send the equilibrium degree of discrimination in the wrong direction, but could also have unintended welfare consequences. 


\vspace{-10pt}

\paragraph{Model discussions.}\label{paragraphmodeldiscussion} Our findings differ from the standard Arrovian mechanism of coordination failure, which obtains discrimination as a Pareto-dominated,  ``bad'' equilibrium for all parties involved (see,  e.g., \citealt{coateloury1993}).  Rational inattention is clearly at work here, because when attention cost is (close to) zero, our game has a unique, impartial, equilibrium that induces both agents to work (recall the diagram in the introduction); it would not feature the coordination failure that is distinctive of Arrow's model of statistical discrimination. Given this benchmark, one can safely attribute all our findings, especially those concerning the discriminatory equilibrium, to rational inattention. 
To ensure that (work, work) is the unique equilibrium when $\lambda \approx 0$, we use Part (ii) of Assumption \ref{assm_regularity}: $c<\bar\mu (1-\bar\mu)/(A+B)$, to bound the effective effort cost from above. Under this condition, together with Part (i) of Assumption \ref{assm_regularity}, all cutpoints $\lambda^*$, $\ul \lambda$, and $\bar\lambda$ are strictly positive. 

The role of tournament as the relevant incentive scheme in our model is also key. Rational inattention turns this competition into a competition for the principal's limited attention, and justifies the use of a discriminatory signal structure in the principal's most preferred equilibrium. If, instead, the principal forms separate contractual relationships with individual agents as in, e.g., \cite{coateloury1993} and \cite{fosgerau2021equilibrium},  then the most profitable equilibrium signal structure between a principal-agent pair is generically unique,  hence discrimination cannot generically arise as the most profitable equilibrium  among  ex-ante identical agents.  

\label{paragraphA1} Additionally, Part (i) of Assumption~\ref{assm_regularity}: $\bar\mu+\ul\mu>1$, ensures that a discriminatory equilibrium emerges under some attention cost parameters. If, instead, $\bar\mu+\ul\mu \leq 1$, then $(X,Y)$ must lie below the 45 degree line in order induce $m$ to work and $w$ to shirk. However, this reliance on the carrot rather than the stick to provide incentives, i.e., $X>Y$, is inconsistent with the optimal discriminatory signal structure, whose primary feature is the stick, i.e., $Y=AX/B>X$. As depicted in Figure~\ref{figure3}, since the red ray on which the optimal signal structure lies doesn't intersect the grey area, no discriminatory equilibrium generically exists when $\bar\mu+\ul\mu \leq 1$.\footnote{$\bar\mu+\ul\mu=1$ happens in the special case where productivity is a noiseless measure of the underlying effort, i.e., $(\bar\mu, \ul \mu)=(1,0)$. This case is worth emphasizing because we use mutual information to measure attention cost. An important property of mutual information (more generally, bounded uniformly separable attention costs) is that information acquisition becomes free at degenerate priors. This property often poses conceptual and technical challenges to the analysis of strategic situations in which players hold endogenous prior beliefs about each other; see  \cite{bloedel2020cost}, \cite{ravid2020ultimatum}, and \cite{denti2022experimental} for discussions of the issue and proposed remedies. 

In our case, the game has a unique, impartial, equilibrium that sustains the high effort profile if $(\bar\mu, \ul\mu)=(1,0)$. The reason is that, under the high effort profile, both agents obtain an expected payoff of $1/2-C>0$.  A  unilateral deviation to low effort will be detected for sure, at zero cost, and reduces the deviator's payoff to zero,  and so is  unprofitable. The proof of why $(\bar\mu, \ul\mu)$ or $(\ul\mu, \ul\mu)$ cannot be sustained in an equilibrium is analogous. } 
\begin{figure}[ht!]
\begin{center}
\scalebox{1}{
\begin{tikzpicture}[scale=.9]
\draw[->] (0,0)--(5.9,0);
\node[right, font=\scriptsize] at (5.9,0){$X$};
\draw[->] (0,0)--(0,5.9);
\node[above, font=\scriptsize] at (0,5.9){$Y$};

\draw[fill=black,opacity=0.1]  (2,2) -- (3.1, 0) -- (5, 0) -- cycle;
\node[right, font=\scriptsize] at (2.7, .4){$IC_{m} \wedge IC_{w}$};
 \draw[-,black, thick] (0, 3.333)--(5,0);
 \draw[-,blue,thick] (0, 5.714)--(3.07,0);
\draw[->] (.8, 5.2)--(.5,5);
\node[black,  font=\scriptsize] at (2.3,  5.2){$(1-\ul\mu) X +
\ul \mu Y=c$};
\draw[->] (3.4,1.45)--(3.1,1.35);
\node[right, font=\scriptsize] at (3.3,1.5){ $\bar\mu X +(1-\bar\mu) Y = c$};

\draw[-,red,thick] (0,0)--(2.3,4.5);
\node[right, font=\scriptsize] at (2.3,4.5){ $Y = AX/B$};

\draw[dashed](0,0)--(2,2);
\node[right, font=\scriptsize] at (0.2,0.2){$45^{\circ}$};
\end{tikzpicture}
}
\caption{No discriminatory equilibrium exists generically when $\bar\mu+
\ul\mu\leq 1$.}\label{figure3}
\end{center}
\vspace{-15pt}
\end{figure}
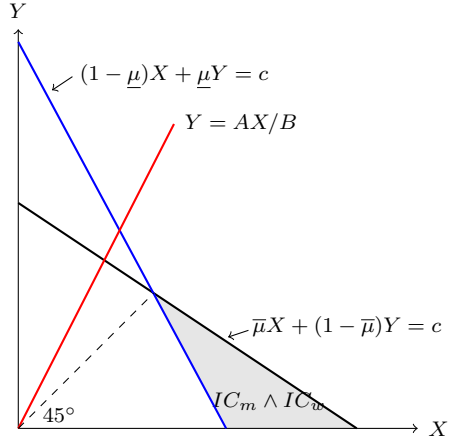

If we maintain Assumption \ref{assm_regularity} but abandon the regularity conditions that ensure $\lambda^*<\bar\lambda$ (stated in Theorem \ref{thm_existenceuniqueness}), then $[\ul\lambda, \bar\lambda] \subsetneq [0, \lambda^*]$. Consequently, the discriminatory equilibrium can only coexist with the impartial, high-effort equilibrium, and it is less profitable than the latter under Assumption \ref{assm_regularity}(i) (see Appendix \ref{sec_proof} for the proof). 

\vspace{-10pt}
\paragraph{Robustness.} In the Online Appendix, we explore the robustness of our model across various dimensions, including alternative attention cost functions, solution concepts, game sequences, agent heterogeneity, and more complex decision-making scenarios for the principal. Our findings indicate that the dependence of the attention cost function on the principal's prior belief is critical. Restricting this aspect of the strategic interaction between the principal and agents diminishes our ability to sustain discrimination as the most profitable equilibrium. \label{referee:costs}

\section{Extension: Affirmative Action Quota}\label{sec_aa}
Our model serves to address various phenomena associated labor market discrimination and to evaluate some policy interventions used to curb discrimination in practice. In this section, we focus on affirmative action quota, which ensure a certain representation of each demographic group. In our model,  this translates into an equal probability of promotion for $m$ and $w$ on average:
\begin{equation}\label{eqn_quota}
\tag{Q} \bar\pi=\frac{1}{2}.
\end{equation} 
The next theorem delineates the channel through which the promotion quota operates.

\begin{theorem}\label{thm_aa}
Under the assumption that $\bar\mu+\ul\mu>1$,  equilibria of the game with quota coincide with the impartial equilibria of the baseline model.
\end{theorem}

Time has not quelled controversy over affirmative action quotas since their introductions in the 1960s and 1970s, with recent studies seeking to understand the channels through which affirmative action quotas operate (see \citealt{holzer2000assessing},  \citealt{fang2011theories}, and \citealt{doleac2021review} for surveys).  Theorem~\ref{thm_aa} adds to this debate, showing that in the current context, the promotion quota operates through eliminating the discriminatory equilibrium of the baseline model without impacting on the impartial equilibria.  More interestingly,  the use of quota does not generate any new equilibrium as a byproduct --- a result that sets our analysis apart from alternative models of Arrovian discrimination such as \cite{coateloury1993}, in which the use of quota may generate new, ``patronizing,'' equilibria whereby the minority group works even less harder than before.

As it turns out, quota operates in our model through effectively subsidizing the principal for hiring the minority.  Technically, it turns the the principal's problem into the following,  holding agents' effort choices fixed:
\begin{align}\label{eqn_aa}
\max_{\pi, a(\cdot)} \mathbb{E} \left[\tilde{a} (\Delta \tilde{\theta}-\nu) \mid  \bm{\mu},\pi, a(\cdot)\right]+\mu_w - \la I(\pi \mid \bm{\mu}).
\end{align}
In the above expression,  the term $\nu$ represents the Lagrange multiplier associated with constraint (\ref{eqn_quota}). It equals zero if agents exert the same level of
effort,  and so the baseline equilibrium is impartial and automatically satisfies
(\ref{eqn_quota}); it is is strictly positive if $\mu_m>\mu_w$,  and so 
(\ref{eqn_quota}) is binding from above; finally it is strictly negative if $\mu_m<\mu_w$, and so (\ref{eqn_quota}) is binding from below. 

Clearly, the use of quota eliminates the discriminatory equilibria of the baseline model without impacting on the impartial equilibria. To show that it does not generate new equilibria (which must induce different levels of effort from the agents because otherwise we are in the impartial case), we provide a partial characterization of the solution to (\ref{eqn_aa}) for any $\nu>0$. 
\begin{lemma}\label{lem_aa}
Fix any $\nu>0$. In the case where the solution to (\ref{eqn_aa}) satisfies (\ref{eqn_quota}), it must also satisfy $\pi(0)<1/2$ and $X>Y>0$. 
\end{lemma}

Lemma \ref{lem_aa} shows that if the principal faces a subsidy for hiring $w$ and happens to promote the agents with equal probability on average,  then he must treat $w$ more favorably unless $m$ is strictly more productive.  Yet such a screening strategy cannot induce $m$ to work and $w$ to shirk (which is needed for $\nu$ to be positive).  \begin{figure}[ht!]
\begin{center}
\scalebox{1}{
\begin{tikzpicture}[scale=.9]
\draw[->] (0,0)--(5.9,0);
\draw[->] (0,0)--(0,5.9);
\node[above, font=\scriptsize] at (0,5.9){$Y$};
\node[right, font=\scriptsize] at (5.9,0){$X$};
\draw[dashed](0,0)--(5,5); 
\node[right, font=\scriptsize] at (0.2,0.2){$45^{\circ}$};
\node[right, font=\scriptsize] at (-.1, 3.6){ $IC_{m}\wedge IC_{w}$};
\node[right, font=\scriptsize] at (3.5,2){ $IC_{m}$};
\node[right, font=\scriptsize] at (1.3,0.5){ $IC_{w}$};
\draw[fill=red,opacity=0.1]  (0,0)  -- (3.07,0) -- (2,2) -- cycle;
\draw[fill=green,opacity=0.1]  (5,5) -- (5,0) -- (2,2) -- cycle;
\draw[fill=black,opacity=0.1]  (2,2) -- (0,3.333) -- (0,5.714) -- cycle;

 \draw[-,blue,thick] (0, 3.333)--(5,0);
 \draw[-,black,thick] (0, 5.714)--(3.07,0);
 \draw[->] (.8, 5.2)--(.5,5);
 \node[black,  font=\scriptsize] at (2.5,  5.3){$\bar\mu X +(1-\bar\mu)Y=c$};
 \draw[->] (4.6,0.6)--(4.4,0.45);
 \node[right, font=\scriptsize] at (4.5,0.5){ $(1-\ul\mu) X +\ul \mu Y = c$};
\end{tikzpicture}
}
\centering
\caption{Under the assumption that $\bar\mu+\ul\mu>1$, a signal structure with $X>Y>0$ cannot satisfy both agents' IC constraints at $\bm\mu=(\bar\mu, \ul\mu)$. }\label{figure4}
\end{center}
\vspace{-15pt}
\end{figure}
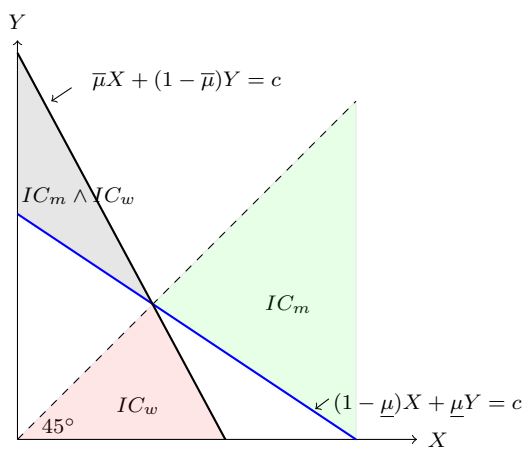
This is illustrated by Figure \ref{figure4},  which gathers all $(X,Y)$ profiles that satisfy both agents' incentive compatibility constraints in the grey area. Under the assumption that $
\bar\mu+\ul\mu>1$, the grey area lies above the 45 degree line and so must not contain the optimal signal structure. The latter is shown to satisfy $X>Y$ and so must lie below the 45 degree line --- an area in which satisfying one agent's incentive compatibility constraint would necessarily violate the incentive compatibility constraint of the other agent (the pink and green areas are disjoint).  The proof for the case where $m$ shirks and $w$ works  (and so $\nu<0$) is analogous and so is omitted.

\section{Related literature}\label{sec_literature}

\paragraph{Statistical discrimination.} 
The literature on statistical discrimination is vast and would be impossible to exhaust here (see the surveys by \citealt{fang2011theories} and \citealt{onuchic2022}). Our work provides a new foundation for the discriminatory information structure assumed by Phelpsian models of statistical discrimination. Recent studies by \cite{chambers2021characterisation}, \cite{escude2022}, and \cite{deb2022wage} examine Phelpsian statistical discrimination from the angle of information design but do not endogenize agents' investment decisions. 

\cite{coateloury1993} develop an Arrovian model of statistical discrimination
with an exogenous, symmetric signal of workers' skills. Like in other traditional Arrovian models, the authors  rely on the complementarity between effort and belief to generate multiple equilibria, and interpret the coexistence of a Pareto-dominant, ``good equilibrium'' and a Pareto-dominated, ``bad equilibrium'' as discrimination across groups. In contrast, our model generates discrimination through the complementarity between effort asymmetry and attention asymmetry. This is achieved by using tournament as the incentive scheme and letting rational inattention turn the competition between agents into a competition for the principal's attention. When the attention cost is intermediate, effort asymmetry and attention asymmetry become mutually reinforcing and can sometimes sustain discrimination as the most profitable equilibrium. To single out the mechanism of our interest, we use regularity conditions to ensure that when attention is (almost) costless, our game has a unique, high-effort equilibrium, hence the traditional Arrovian mechansim is absent from our model.  \label{pagecoateloury}

 The role of information acquisition in shaping discrimination has recently been explored by \cite{bartovs2016attention} and \cite{fosgerau2021equilibrium}. The model of \cite{bartovs2016attention} takes as given the exogenous differences between groups, as well as employers' default decisions regarding whether to accept or reject minorities  absent information acquisition (lemon dropping vs. cherry-picking). Here, workers' investment decisions and the employer's choice of signal structure are mutually enforcing. This additional source of endogeneity raises the possibility of sustaining discrimination among ex-ante identical workers, and predicts a nonmonotonic relation between the equilibrium degree of discrimination and the cost of information acquisition.\footnote{The last result also separates our analysis from earlier works that combine exogenous, asymmetric signals with endogenous Arrovian investments \citep{borjas1978biased,lundberg1983private}.  }

\cite{fosgerau2021equilibrium} study an Arrovian model where a screener incurs a general posterior-separable attention cost to acquire information about job candidates.  Since candidates are screened on an individual basis, discrimination cannot generically arise as the most preferred equilibrium by the screener for reasons discussed in Section \ref{sec:implications}. The focus is instead on how rational inattention interacts with natural, intrinsic, differences between groups such as prejudice and asymmetric access to social capital.

\vspace{-10pt}
\paragraph{Rational inattention.} The literature on rational inattention (RI) pioneered by \cite{sims2003implications} has grown substantially in recent years; see 
\cite{mackowiak2023rational} for a survey. We use the ideas and techniques developed
in this literature to study statistical discrimination. Conceptually,  our results exploit the flexibility associated with RI information acquisition. The link between attentional flexibility and discrimination has long been recognized and documented by psychologists \citep{eberhardt2020biased} and more recently by economists \citep{bartovs2016attention, huang2022, glover2017discrimination}.\footnote{In other areas of economics, attentional flexibility has proven crucial for shaping the outcomes of financial contracting, political competition, and ultimatum bargaining \citep{yang2020optimality, hu2019politics, ravid2020ultimatum}. Its  empirical relevance has been established by \cite{dean2023experimental}.} Technically, 
\cite{matvejka2015rational}  and \cite{yang2020optimality}  provide a complete characterization of the optimal signal structure for binary decision problems, while \cite{matveenko2021attentional} study how imposing quotas on the average decision probabilities affects the solution to the RI decision problem studied by \cite{matvejka2015rational}.  Our analysis of the principal's problem builds on their results.

\vspace{-10pt}

\paragraph{Incentive contracting.} Since \cite{alchian1972production}, there has been a long tradition of studying the role of monitoring cost in shaping the organization of principal-agent relationships.  \cite{liyang} examine the problem faced by a rationally inattentive principal who can simultaneously design the monitoring technology and incentive scheme as a package, while restricting attention to partitional monitoring technologies. Here the incentive scheme is taken as exogenously given but the monitoring technology is unrestricted. 
    
The theory of contests has been used to inform affirmative action policies that level the playing field for heterogeneous participants \citep{chowdhury2020heterogeneity}, while recent works by \cite{drugov2017biased} and \cite{fu2020optimal} show that the optimal contest between symmetric agents can be biased in sufficiently general environments. We focus on the role of RI in biasing the equilibrium signal structure and investment decisions, embedding the analysis in a simple contest game between symmetric agents. 

\appendix

\section{Proofs}\label{sec_proof}

\subsection{Preliminaries}\label{sec_proof_preliminary}
\paragraph{Notations.} Throughout this appendix, we follow the notational conventions developed in the main text. Specifically, we use $\bm{\mu}\coloneqq (\mu_m, \mu_w)$ denote the profile of effort choices by $m$ and $w$,  and $\Delta\theta \in \{-1,0,1\}$ to denote the differential productivity between them. For any signal structure $\pi$, we write $X$ for $\pi(1)-\pi(0)$ and $Y$ for $\pi(0)-\pi(-1)$. Finally, recall the following definitions: $\Delta \mu \coloneqq \bar\mu-\ul\mu$, $c \coloneqq C/\Delta\mu$, $A\coloneqq \bar\mu(1-\ul\mu)$, $B \coloneqq \ul\mu(1-\bar\mu)$, and $\gamma \coloneqq \exp(1/\lambda)$. Note that $A>B$, and that $\gamma$ is decreasing in $\lambda$ and satisfies $\gamma \rightarrow +\infty$ as $\lambda\rightarrow 0$ and $\gamma \rightarrow 1$ as $\lambda \rightarrow +\infty$. 

In addition, define $h(x)\coloneqq x\ln x+ (1-x) \ln (1-x)$ $\forall x \in [0,1]$ and note that $-h$ is the binary entropy function. $h$ equals zero at $x=0$ and $x=1$; it is decreasing on $[0,1/2]$ and is increasing on $[1/2,1]$. 

\vspace{-8pt}

\paragraph{Binary RI decision problem.} Fix any effort profile $\bm{\mu} \in \{\ul\mu, \bar\mu\}^2$, and let $p(\Delta \theta)$ denote the probability that $\Delta \theta$ occurs under $\bm{\mu}$. For any signal structure $\pi$, $\bar\pi \coloneqq \sum_{\Delta \theta \in \{-1,0,1\}}p(\Delta \theta)\pi(\Delta\theta)$ is the average probability that $\pi$ recommends $m$ for promotion, and $1-\bar\pi$ is the average probability that it recommends $w$ for promotion. The expected revenue generated by $\pi$ equals: \[\sum_{\Delta \theta \in \{-1,0,1\}} p(\Delta \theta)\pi(\Delta \theta)\Delta\theta + \mu_w,\]
and the mutual information cost incurred by $\pi$ equals: \[\sum_{\Delta \theta \in \{-1,0,1\}} p(\Delta \theta)h(\pi(\Delta\theta))-h(\bar\pi).\] 

The principal problem: (\ref{eqn_mainproblem}), is to choose a signal structure that maximizes the expected revenue minus $\lambda$ times the mutual information cost. By Proposition 1 of \cite{yang2020optimality},  the solution to this problem --- which we denote simply by $\pi$ --- uniquely exists and satisfies $\pi(\Delta \theta) \equiv 1$  if $\mathbb{E}[\exp(-\Delta \ta/\la) \mid \bm{\mu}] \leq 1$, $\pi (\Delta \theta ) \equiv 0$ if $\mathbb{E}[\exp(\Delta \ta/\la) \mid \bm{\mu}] \leq 1$, and $\pi(\Delta \theta) \in (0,1)$
 $\forall \Delta \theta$ otherwise. Simplifying the last condition yields: $\forall \Delta \theta \in \{-1,0,1\}$, 
\begin{equation}\label{eqn_yangcondition}
  \pi(\Delta\ta) 
  \begin{cases}
    =  1 & \text{ if } p(1)/p(-1) \geq \gamma, \\
    =0 & \text{ if } p(1)/p(-1) \leq 1/\gamma, \\
    \in (0,1) & \text{ else}.  
  \end{cases}
  \end{equation}
In what follows, we say that $\pi$ is degenerate in the first two cases, and that it is nondegenerate in the last case. 

When nondegenerate, $\pi$ satisfies the multinomial logit formula prescribed by Theorem 1 of \cite{matvejka2015rational}: 
\begin{equation}\label{eqn_logit}
\pi(\Delta \theta)=\frac{\bar{\pi} \exp(\Delta \theta/\lambda)}{\bar\pi \exp(\Delta \theta/\lambda)+1-\bar\pi} \quad \forall \Delta \theta,
\end{equation}  
and $\bar\pi$ must satisfy Bayes' plausibility: 
\begin{equation}\label{eqn_pibar}
\sum_{\Delta \theta \in \{-1, 0, 1\}} p(\Delta \theta) \pi(\Delta \theta)=\bar{\pi}. 
\end{equation}
(\ref{eqn_yangcondition}), (\ref{eqn_logit}), and (\ref{eqn_pibar}) together fully pin down $\pi$. 


\subsection{Useful lemmas and their proofs}\label{sec_proof_lemma}
\paragraph{Proof of Lemma \ref{lem_optimalsignal}.}  Part (i): When $\bm{\mu}=(\bar\mu, \bar\mu)$, we have $p(1)=p(-1)=\bar{\mu} (1-\bar{\mu})$ and so  $p(1)/p(-1)=1 \in (1/\gamma, \gamma)$. Consequently, $\pi$ is always nondegenerate and is fully pinned down by (\ref{eqn_logit}) and (\ref{eqn_pibar}). Solving $\pi$ explicitly yields: 
\[\bar\pi=\pi(0)=\frac{1}{2} \text{ and } X=Y=g(\gamma) \coloneqq \frac{\gamma-1}{2(\gamma+1)},\]
where the function $g:[1, +\infty) \rightarrow \mathbb{R}$ satisfies $g>0$ and $g'>0$ $\forall \gamma>1$. Since $\gamma \coloneqq \exp(1/\lambda)$ is decreasing in $\lambda$, the last result can be expressed as $dg(\gamma)/d\lambda<0$ $\forall \lambda>0$. The proof for the case of  $\bm{\mu}=(\ul\mu, \ul\mu)$ is analogous and so is omitted. 
\vspace{8pt}

\noindent Part (ii): When $\bm{\mu}=(\bar \mu, \ul \mu)$, we have $p(1)=A$ and $p(-1)= B$, hence $p(1)/p(-1)=A/B>1$. Consequently, $p(1)/p(-1)<1/\gamma$ can never happen, while $p(1)/p(-1) \geq \gamma$ holds if and only if \[\gamma \leq \breve{\gamma}\coloneqq \frac{A}{B}, \text{ or equivalently } \lambda \geq \breve{\lambda}\coloneqq (\ln\breve{\gamma})^{-1}.\] For all $\lambda<\breve{\lambda}$, $\pi$ is nondegenerate and is fully pinned down by (\ref{eqn_logit}) and (\ref{eqn_pibar}). Solving $\pi$ explicitly yields:
\[\bar\pi =\pi(0)=\frac{\gamma A- B}{(\gamma-1)(A+B)}, \text{ }  X=f(\gamma)\coloneqq \frac{(\gamma A-B)(\gamma B-A)}{(\gamma^2-1)(A+B)A}, \text{ and } Y=\frac{A}{B}f(\gamma),\]
where the function $f: [\breve{\gamma},+\infty) \rightarrow \mathbb{R}$ satisfies $f>0$ and $f'>0$ $\forall \gamma>\breve\gamma$ (equivalently, $df(\gamma)/d\lambda<0$ $\forall \lambda<\breve{\lambda}$). The proof for the case of $\bm{\mu}=(\ul\mu, \bar\mu)$ is analogous and so is omitted.  \qed
\vspace{-8pt}
\paragraph{Proof of Lemma \ref{lem_optimaleffort}.}
For any given $\mu_w$ and $\pi$, $m$ prefers to exert high effort rather than low effort if and only if 
\begin{multline*}
\bar\mu(1-\mu_w)\pi(1)+(1-\bar\mu)\mu_w \pi (-1)+[1-\bar\mu(1-\mu_w)-(1-\bar\mu)\mu]\pi(0)-C\\
\geq \ul\mu(1-\mu_w)\pi(1)+(1-\ul\mu)\mu_w \pi (-1)+[1-\ul\mu(1-\mu_w)-(1-\ul\mu)\mu]\pi(0),
\end{multline*}
or, equivalently,  
\[(1-\mu_w)X+\mu_w Y \geq c\coloneqq \frac{C}{\Delta \mu}.  \]
Likewise, $w$ prefers to exert high effort rather than low effort if and only if 
\begin{multline*}
(1-\bar\mu)\mu_m\pi(1)+\bar\mu(1-\mu_m) \pi (-1)+[1-(1-\bar\mu)\mu_m-\bar\mu(1-\mu_m)]\pi(0)-C\\
\geq (1-\ul\mu)\mu_m\pi(1)+\ul\mu(1-\mu_m) \pi (-1)+[1-(1-\ul\mu)\mu_m-\ul\mu(1-\mu_m)]\pi(0),
\end{multline*}
or, equivalently,  \[\mu_m X+(1-\mu_m)Y \geq c.\qed \]

\begin{lemma}\label{lem_algebra}
Let $V(\bm{\mu};\gamma)$ and $I(\bm{\mu};\gamma)$ denote the expected revenue and the mutual information cost generated by the optimal signal structure for $\bm{\mu}$, respectively, when the attention cost parameter is $(\ln\gamma)^{-1}$.  
Then 
\begin{align*}
&V((\bar\mu, \bar{\mu});\gamma)=\bar \mu+\bar\mu(1-\bar\mu)\frac{\gamma-1}{\gamma+1},  \text{  } V((\bar{\mu}, \ul\mu);\gamma)=\ul{\mu}+\frac{\gamma A-B}{\gamma+1}, \\
&V((\ul\mu, \ul{\mu});\gamma)=\ul\mu+\ul\mu(1-\ul \mu) \frac{\gamma-1}{\gamma+1}, \text{ } V((\bar\mu, \bar\mu);\gamma)-V((\bar\mu, \ul\mu);\gamma)=\frac{\Delta \mu}{\gamma+1}[\gamma-(\gamma-1)\bar\mu],\\
&V((\bar\mu, \ul\mu);\gamma)-V((\ul\mu, \ul\mu);\gamma)=\frac{\Delta \mu}{\gamma+1}[\gamma-(\gamma-1)\ul\mu],\\
&\frac{d}{d\gamma}V((\bar\mu, \bar\mu);\gamma)-V((\bar \mu, \ul \mu); \gamma)=\frac{\Delta \mu (1-2\bar\mu)}{(\gamma+1)^2},  \text{ and}\\
&\frac{d}{d\gamma}V((\bar\mu, \ul\mu);\gamma)-V((\ul\mu, \ul \mu); \gamma)=\frac{\Delta \mu (1-2\ul\mu)}{(\gamma+1)^2}, 
\end{align*}
and 
\begin{align*}
&I((\bar\mu, \bar{\mu});\gamma)=2\bar\mu(1-\bar\mu)[h\left(\frac{\gamma}{\gamma+1}\right)-h\left(\frac{1}{2}\right)],\\
&I((\bar{\mu}, \ul\mu);\gamma)=A h\left(\frac{\gamma(\gamma A-B)}{(\gamma^2-1)A}\right)+Bh\left(\frac{\gamma A-B}{(\gamma^2-1)B}\right)-(A+B)h\left(\frac{\gamma A-B}{(\gamma-1)(A+B)}\right),\\
&I((\ul\mu, \ul{\mu});\gamma)=2\ul\mu(1-\ul\mu)[h\left(\frac{\gamma}{\gamma+1}\right)-h\left(\frac{1}{2}\right)], \\
&\frac{d}{d\gamma} I((\bar\mu, \bar\mu);\gamma)-I((\bar \mu, \ul \mu); \gamma)=\frac{\Delta \mu (1-2\bar\mu)\ln \gamma}{(\gamma+1)^2},\text{and } \\
&\frac{d}{d\gamma} I((\bar\mu, \ul\mu);\gamma)-I((\ul \mu, \ul \mu); \gamma)=\frac{\Delta \mu (1-2\ul\mu)\ln \gamma}{(\gamma+1)^2}.
\end{align*}
\end{lemma} 

\begin{proof}
When proving Lemma \ref{lem_optimalsignal}, we solved for the optimal signal structure for any given $\bm{\mu}$. Substituting the solutions into the expressions for $V(\cdot;\gamma)$ and $I(\cdot;\gamma)$ gives the results concerning $V$ and $I$, as well as their derivatives w.r.t. $\gamma$. $dV(\cdot; \gamma)/d\gamma$ is easy to calculate. Meanwhile, 
\[\frac{d}{d\gamma}I((\bar\mu, \bar\mu); \gamma)=\frac{2\bar\mu(1-\bar\mu)\ln \gamma}{(\gamma+1)^2}, \text{ } \frac{d}{d\gamma}I((\ul\mu, \ul\mu); \gamma)=\frac{2\ul\mu(1-\ul\mu)\ln \gamma}{(\gamma+1)^2},\]
\[\text{ and } \frac{d}{d\gamma}I((\bar\mu, \ul\mu); \gamma)=\frac{(A+B)\ln \gamma}{(\gamma+1)^2}.\]
The first two derivatives are easy to calculate, while the last one requires a bit of work. 

For starters, recall that the optimal signal structure for $(\bar\mu,\ul\mu)$ satisfies:
\begin{align*}
&\bar\pi=\frac{\gamma A-B}{(\gamma-1)(A+B)}, \text{ } 1-\bar\pi=\frac{\gamma B}{(\gamma-1)(A+B)}, \text{ and } \frac{\bar\pi}{1-\bar\pi}=\frac{\gamma A-B}{\gamma B-A};\\
& \pi(1)=\frac{\gamma(\gamma A-B)}{(\gamma^2-1)A}, \text{ } 1-\pi(1)=\frac{\gamma B-A}{(\gamma^2-1)A}, \text{ and } \frac{\pi(1)}{1-\pi(1)}=\frac{\gamma(\gamma A-B)}{\gamma B-A};\\
\text{ and } & \pi(-1)=\frac{\gamma A-B}{(\gamma^2-1)B}, \text{ } 1-\pi(-1)=\frac{\gamma (\gamma B-A)}{(\gamma^2-1)B}, \text{ and } \frac{\pi(-1)}{1-\pi(-1)}=\frac{\gamma A-B}{\gamma(\gamma B-A)}. 
\end{align*}
Differentiating $\bar\pi$, $\pi(1)$, and $\pi(-1)$ w.r.t. $\gamma$ yields:
\[\frac{d\bar\pi}{d\gamma}=\frac{-(A-B)}{(A+B)(\gamma-1)^2}, \text{ }\frac{d\pi(1)}{d\gamma}=\frac{B(\gamma^2+1)-2A\gamma}{A(\gamma^2-1)^2},\text{ and } \frac{d\pi(-1)}{d\gamma}=\frac{-A(\gamma^2+1)+2B\gamma}{B(\gamma^2-1)^2}. \]
Based on these results, alongside with $h'(x)=\ln\left(\frac{x}{1-x}\right)$,  we obtain that  
\begin{align*}
&\frac{d I((\bar\mu, \ul\mu);\gamma)}{d\gamma}\\
&=A\ln\left(\frac{\pi(1)}{1-\pi(1)}\right)\frac{d\pi(1)}{d\gamma}+B\ln\left(\frac{\pi(-1)}{1-\pi(-1)}\right)\frac{d\pi(-1)}{d\gamma}-(A+B)\ln\left(\frac{\bar\pi}{1-\bar\pi}\right)\frac{d\bar\pi}{d\gamma}\\
&=A \frac{B(\gamma^2+1)-2A\gamma}{A(\gamma^2-1)^2}\ln\left(\frac{\gamma (\gamma A-B)}{\gamma B-A}\right)+B \frac{-A(\gamma^2+1)+2B\gamma}{B(\gamma^2-1)^2}\ln\left(\frac{\gamma A-B}{\gamma (\gamma B-A)}\right)\\
& \quad -(A+B)\frac{-1}{(\gamma-1)^2}\frac{A-B}{A+B}\ln\left(\frac{\gamma A-B}{\gamma B-A}\right)\\
&=\left[\frac{B(\gamma^2+1)-2A\gamma -A(\gamma^2+1)+2B\gamma}{(\gamma^2-1)^2} +\frac{A-B}{(\gamma-1)^2}\right]\ln \left(\frac{\gamma A-B}{\gamma B-A}\right)\\
& \quad + \frac{B(\gamma^2+1)-2A\gamma+A(\gamma^2+1)-2B\gamma}{(\gamma^2-1)^2}\ln \gamma =0+\frac{(A+B)\ln\gamma}{(\gamma+1)^2},
\end{align*}
as desired. 
\end{proof}

\vspace{-8pt}

\paragraph{Proof of Lemma \ref{lem_costrank}.} By Lemma \ref{lem_algebra}, the following holds for all $\gamma>1$:
\[
I((\bar\mu, \bar{\mu});\gamma)-I((\ul\mu, \ul{\mu});\gamma)=-2\Delta\mu(\bar\mu+\ul\mu-1)[h\left(\frac{\gamma}{\gamma+1}\right)-h\left(\frac{1}{2}\right)]< 0 , 
\]
where the last inequality uses Part (i) of Assumption \ref{assm_regularity}: $\bar\mu+\ul\mu>1$,  and the fact that $h$ is increasing on $[1/2,1]$. Consequently, the mutual information cost incurred by the high effort profile is smaller than that incurred by the low effort profile. 

If, in addition, $\ul\mu>1/2$ (as required by Theorems \ref{thm_existenceuniqueness}(iii) and \ref{thm_mostprofitable}), then 
\[\frac{d}{d\gamma} I ((\bar\mu, \ul\mu); \gamma)-I ((\ul\mu, \ul\mu); \gamma)=\frac{\Delta\mu (1-2\ul\mu)\ln\gamma}{(\gamma+1)^2}<0 \text{ }\forall \gamma \in [\breve{\gamma}, +\infty). \]
Meanwhile,
\[I ((\bar\mu, \ul\mu); \breve{\gamma})-I ((\ul\mu, \ul\mu); \breve{\gamma})=0-2\ul\mu(1-\ul\mu)[h\left(\frac{\breve\gamma}{\breve\gamma+1}\right)-h\left(\frac{1}{2}\right)]<0,\]
where the equality exploits the definition of $\breve{\gamma}$, and the inequality exploits the fact that $\breve\gamma>1$ and $h$ is increasing on $[1/2,1]$. Combining these results, we obtain that 
 $I ((\bar\mu, \ul\mu); \gamma)<I ((\ul\mu, \ul\mu); \gamma)$ $\forall \gamma\geq \breve\gamma$. That is, the mutual information cost incurred by the discriminatory effort profile is smaller than that incurred by the low effort profile whenever the optimal discriminatory signal structure is nondegenerate. \qed 

 \vspace{-8pt}

\paragraph{Proof of Lemma \ref{lem_aa}.}  Fix any $\nu > 0$. A careful inspection reveals that problem (\ref{eqn_aa}): 
\[\max_{\pi, a(\cdot)}\mathbb{E} \left[\tilde{a} (\Delta \tilde{\theta}-\nu) \mid  \bm{\mu},\pi, a(\cdot)\right]+\mu_w- \la I(\pi \mid \bm{\mu}),
\]
is nothing but the very kind of the RI decision problem described in Appendix \ref{sec_proof_preliminary}, whereby the principal's payoff difference from choosing $m$ over $w$  is $\Delta \tilde{\theta}-\nu$ rather than $\Delta\tilde{\theta}$.  Modifying (\ref{eqn_logit}) accordingly yields: 
\[\pi(\Delta \theta)=\frac{\bar{\pi}
\exp((\Delta \theta-\nu)/\lambda)}{\bar\pi\exp((\Delta\theta-\nu)/\lambda)+1-\bar\pi} \quad \forall \Delta. \theta\]
In the case where $\bar\pi=1/2$,  the above expression simplifies to: 
\[\pi(\Delta \theta)=\frac{\exp((\Delta \theta-\nu)/\lambda)}{\exp((\Delta\theta-\nu)/\lambda)+1} \quad \forall \Delta \theta,\]
so in particular $\pi(0)<1/2$.  Further algebra shows that 
\[X=\frac{\exp(1/\la)-1}{[\exp((1-\nu)/\la)+1][\exp(\nu/\la)+1]}\text{ and } Y=\frac{\exp(\nu/\la)(\exp(1/\la)-1)}{[\exp(\nu/\la)+1][\exp((\nu+1)/\la)+1]} \]
where
\[\frac{X}{Y}=\frac{\exp((\nu+1)/\la)+1}{\exp(1/\la)+\exp(\nu/\la)}>1\]
due to the convexity of the exponential function. \qed

\subsection{Proofs of theorems and propositions}\label{sec_proof_theorem}
\paragraph{Proof of Theorem \ref{thm_existenceuniqueness}.} Under Assumption \ref{assm_regularity}(i), i.e., $\bar\mu+\ul\mu>1$, the following must hold: 
\[\frac{\bar\mu(1-\bar\mu)}{A+B}-\frac{1}{2}=\frac{\Delta \mu (1-2\bar\mu)}{2(A+B)}<0 \text{ and }\frac{\bar\mu(1-\bar\mu)}{A+B}-\frac{\ul\mu(1-\ul\mu)}{A+B}=\frac{\Delta\mu(1-\bar\mu-\ul\mu)}{A+B}<0.\]
Combining with Assumption \ref{assm_regularity}(ii) yields: 
\[c<\frac{\bar\mu(1-\bar\mu)}{A+B}<\min\{\frac{1}{2}, \frac{\ul\mu(1-\ul\mu)}{A+B}\},\]
a condition that will be invoked extensively in the upcoming proof. 

\vspace{8pt}
\noindent Part (i): Lemma \ref{lem_optimalsignal}(i) and Lemma \ref{lem_optimaleffort} together imply that $(\bar{\mu}, \bar{\mu})$ can be sustained in an equilibrium if and only if $g(\gamma) \geq c$. Since $g(1)=0$, $g'>0$ on $(1, +\infty)$, and $\lim_{\gamma \rightarrow +\infty}g(\gamma)=1/2>c$, $g(\gamma) \geq c$ holds if and only if \[\gamma \geq \gamma^* \coloneqq g^{-1}(c), \text{ or equivalently } \lambda \leq (\ln \gamma^*)^{-1} \coloneqq \lambda^*>0.\]
When the last condition fails, we have $g(\gamma)<c$ and so can sustain $(\ul\mu, \ul\mu)$ can in an equilibrium. At $\gamma=\gamma^*$ (equivalently, $\lambda=\lambda^*$), both $(\bar\mu, \bar\mu)$ and $(\ul \mu, \ul\mu)$ can be sustained in an equilibrium. 
\vspace{5pt}

\noindent Part (ii):  $(\bar\mu, \ul\mu)$ can be sustained in an equilibrium if and only the optimal signal structure for $(\bar\mu, \ul\mu)$ --- which satisfy (i) $X=f(\gamma)$ and (ii) $Y=AX/B$ --- also satisfy (iii) agents' incentive compatibility constraints: $(1-\ul\mu)X+\ul\mu Y \geq c$ and $\bar\mu X+(1-\bar\mu)Y \leq c$.  Solving (ii) and (iii) simultaneously yields $X \in [\ul X, \bar X]$, where 
\[\ul X\coloneqq \frac{c(1-\bar{\mu})}{1-\ul{\mu}} \text{ and } \bar X\coloneqq \frac{c\ul \mu}{\bar \mu}. \]
Note that $\ul X$ and $\bar X$ are both independent of $\gamma$. Moreover, $\ul{X}<B/(A+B)$ because 
\[\text{Assumption \ref{assm_regularity}} \Longrightarrow c<\frac{\ul \mu (1-\ul \mu)}{A+B} \Longrightarrow \ul{X}<\frac{B}{A+B},\]
and $\bar X<B/(A+B)$ because 
\[\bar X =\frac{c\ul\mu}{\bar\mu}<\frac{\bar\mu(1-\bar\mu)}{A+B}\frac{\ul\mu}{\bar\mu}=\frac{B}{A+B}.\]
Then from $f'>0$ $\forall \gamma \in (\breve\gamma, +\infty)$, $f(\breve\gamma)=0$, and $\lim_{\gamma \rightarrow +\infty}f(\gamma)=B/(A+B)$,  it follows that (i) holds if and only if $\gamma \in [\ul\gamma, \bar\gamma]$,
where 
\[\ul \gamma \coloneqq f^{-1}(\ul X) \text{ and }\bar \gamma\coloneqq f^{-1}(\bar{X})\]
are both finite and greather than  $\breve{\gamma}$. Define 
\[\ul \lambda \coloneqq (\ln \bar{\gamma})^{-1} \text{ and } \bar \lambda \coloneqq (\ln \ul{\gamma})^{-1},\]
and note that $0<\ul\lambda<\bar\lambda<\breve{\lambda}<+\infty$. 

It remains to show that $\ul\lambda<\lambda^*$ (equivalently, $\gamma^*<\bar\gamma$) always holds, and that $\lambda^*<\bar\lambda$ (equivalently, $\ul\gamma<\gamma^*$) holds under additional conditions. To prove the first claim, rewrite $f(\gamma)=\bar{X}$ as 
\[\varphi(\gamma)\coloneqq \frac{(\gamma A-B)(\gamma B-A)}{(\gamma^2-1)\ul\mu(1-\ul \mu)(A+B)}=c, \]
where $\varphi: [\breve\gamma, +\infty)\rightarrow  \mathbb{R}$ satisfies $\varphi(\breve\gamma)=0$ and $\varphi'>0$ $\forall \gamma>\breve\gamma$. Consequently, 
$\bar \gamma$ is the unique root of $\varphi(\gamma)=c$, while $\gamma^*$ is the unique root of $g(\gamma)=c$, with $g:[1, +\infty) \rightarrow \mathbb{R}$ satisfying $g(1)=0$ and $g'>0$ $\forall \gamma>1$. Tedious algebra shows that 
\[
\frac{d}{d\gamma}\frac{\varphi(\gamma)}{g(\gamma)}=\frac{2(A-B)^2(\gamma+1)}{\ul\mu(1-\ul \mu)(A+B)(\gamma-1)^3}>0
\]
and that 
\[
\lim_{\gamma \rightarrow +\infty}\varphi(\gamma)=\frac{\bar \mu (1-\bar \mu)}{A+B}<\frac{1}{2}=\lim_{\gamma \rightarrow +\infty} g(\gamma).\]
Therefore, $\varphi(\gamma)<g(\gamma)$ $\forall \gamma \in [\breve\gamma, +\infty)$, hence $\gamma^*<\bar{\gamma}$ must hold.

To pin down the conditions for $\ul\gamma<\gamma^*$ to hold, rewrite $f(\gamma)=\ul X$ as
\[\psi(\gamma)\coloneqq \frac{\ul\mu(1-\ul\mu)}{\bar\mu(1-\bar\mu)}\varphi(\gamma)=c,\]
and $\ul\gamma$ as the unique root of $\psi(\gamma)=c$. From the above derivation, we deduce that 
\[\frac{d}{d\gamma}\frac{\psi(\gamma)}{g(\gamma)}>0\]
and that 
\[\lim_{\gamma \rightarrow +\infty}\psi(\gamma)-\lim_{\gamma \rightarrow +\infty} g(\gamma)=\frac{\ul\mu(1-\ul\mu)}{A+B}-\frac{1}{2}=\frac{\Delta \mu (2\ul\mu-1)}{2(A+B)}.\]
Thus $\gamma^*>\ul\gamma$ if and only if
\begin{equation}\label{eqn_regularity}
\ul\mu>\frac{1}{2} \text{ and }c>g(\hat{\gamma}),
\end{equation}where $\hat{\gamma}$ denotes the unique root of $g(\gamma)=\psi(\gamma)$. Numerical analysis confirms that (\ref{eqn_regularity}) can hold simultaneously with Assumption \ref{assm_regularity}. \qed

\vspace{-8pt}

\paragraph{Proof of Theorem \ref{thm_mostprofitable}.} We only show that the impartial, high-effort equilibrium is more profitable than the discriminatory equilibrium on $[\ul\lambda, \lambda^*]$ (equivalently, $ [\gamma^*,\bar\gamma]$).  The other claims in the theorem follow immediately from Lemma \ref{lem_costrank}. 

Write $\Delta V(\gamma)$ for $V((\bar\mu,\bar\mu); \gamma)-V((\bar\mu, \ul\mu);\gamma)$,  $\Delta I (\gamma)$ for $I((\bar\mu,\bar\mu); \gamma)-I((\bar\mu, \ul\mu);\gamma)$, and $\Delta R(\gamma)$ for $\Delta V(\gamma)-\Delta I(\gamma)/\ln\gamma$. Consider the extended domain $[\breve{\gamma},+\infty)$ (in that it nests the relevant domain $[\gamma^*, \bar\gamma]$), where $\breve{\gamma}>1$. Since 
\[\tag{$\because \breve\gamma>1$ and $\argmin_{[0,1]} h=1/2$}\Delta I\left(\breve\gamma\right)=2\bar\mu(1-\bar\mu)[h\left(\frac{\breve\gamma}{\breve \gamma+1}\right)-h\left(\frac{1}{2}\right)]-0>0\]
and 
\[\tag{$\because \bar{\mu}>1/2$ by Assumption \ref{assm_regularity}(i)}\frac{d}{d\gamma}\Delta I(\gamma)=\frac{\Delta \mu (1-2\bar\mu)\ln \gamma}{(\gamma+1)^2}<0,\]
either $\Delta I(\gamma)>0$ $\forall \gamma \in [\breve\gamma, +\infty)$, or it single crosses the horizontal line from above at some $\tilde\gamma >\breve\gamma$. Then from 
\begin{align*}
\frac{d}{d\gamma}\Delta R(\gamma)
&=\frac{d}{d\gamma} [\Delta V(\gamma)-\frac{1}{\ln\gamma}\Delta I(\gamma)]\\
&=\frac{d\Delta V(\gamma)}{d\gamma}-\frac{1}{
\ln \gamma}\frac{d\Delta I(\gamma)}{d\gamma}+\frac{\Delta I(\gamma)}{\gamma (\ln\gamma)^2}\\
\tag{$\because $ Lemma \ref{lem_algebra}}&= \cancel{\frac{\Delta \mu (1-2\bar\mu)}{(\gamma+1)^2}}-\cancel{\frac{1}{\ln \gamma}\frac{\Delta \mu (1-2\bar\mu)\ln \gamma}{(\gamma+1)^2}}+\frac{\Delta I(\gamma)}{\gamma (\ln\gamma)^2},
\end{align*}
it follows that $\Delta R(\gamma)$ is either monotonically increasing on $[\breve\gamma, +\infty)$, or it first increases on $[\breve\gamma, \tilde{\gamma}]$ and then decreases on $(\tilde{\gamma}, +\infty)$. In both situations, we have 
\[\lim_{\gamma \rightarrow +\infty} \Delta R(\gamma)=\lim_{\gamma \rightarrow +\infty} \Delta V(\gamma)-0\cdot \lim_{\gamma\rightarrow +\infty} \Delta I(\gamma)=\Delta \mu(1-\bar\mu)-0>0.\]
Thus if $\Delta R(\breve\gamma)>0$, then $\Delta R(\gamma)>0$ $\forall \gamma \in [\breve\gamma, +\infty)$ as desired. 

To show that  $\Delta R(\breve\gamma)>0$, note that $V((\bar\mu, \ul\mu); \breve\gamma)=\ul\mu$ by Lemma \ref{lem_algebra}, and that $I(\bar\mu, \ul\mu); \breve\gamma)=0$ by the definition of $\breve{\gamma}$. Also note that $V((\bar\mu, \bar\mu); \breve\gamma)-I((\bar\mu, \bar\mu); \breve\gamma)/\ln \breve\gamma\geq \bar\mu$, where $\bar\mu$ is the expected profit  generated by $(\bar\mu, \bar\mu)$ if the principal uses a degenerate signal structure that recommends $m$ for promotion for sure. The inequality follows from optimality, i.e., the optimal signal structure for $(\bar\mu, \bar\mu)$ generates a (weakly) higher expected profit to the principal than the aforementioned degenerate signal structure. Taken together, we conclude that $\Delta R(\breve\gamma)>\Delta \mu>0$ as conjectured. \qed 
\vspace{-8pt}

\paragraph{Proof of Proposition \ref{prop_welfare}.} The part concerning the principal's welfare follows immediately from Theorem \ref{thm_mostprofitable}. It is also evident that both agents prefer the impartial, low-effort equilibrium in which their expected payoff equals $1/2$, to the impartial, high-effort equilibrium in which their expected payoff equals $1/2-C$. What remains to be verified is that $m$ most prefers the discriminatory equilibrium sustaining $(\bar\mu,\ul\mu)$, i.e., $\bar\pi-C>1/2$, while $w$ least prefers the discriminatory equilibrium sustaining $(\bar\mu,\ul\mu)$, i.e., $1-\bar\pi<1/2-C$. Both results hold if and only if $\bar\pi-1/2>C$, where $\bar\pi$ denotes the average probability that $m$ is recommended for promotion under $(\bar\mu, \ul\mu)$ and is shown to equal $(\gamma A-B)/[(\gamma-1)(A+B)]$ by Lemma \ref{lem_optimalsignal}. 

Subtracting $1/2$ from $\bar\pi$ and doing straightforward algebra, we obtain that 
\[\bar\pi-\frac{1}{2}=\frac{(\gamma+1)\Delta\mu}{2(\gamma-1)(A+B)}>\frac{\Delta\mu}{2(A+B)},\]
where the equality uses the fact that $A-B=\Delta \mu$, and the inequality the fact that $\gamma>1$. Under Assumption \ref{assm_regularity}(ii):  $C/\Delta \mu<\bar\mu(1-\bar\mu)/(A+B)$, the last expression is greater than $C$ because $\bar\mu(1-\bar\mu)<1/4$.  \qed

\vspace{-8pt}

\paragraph{Proof of Theorems \ref{thm_aa}.} 
It is evident that the use of quota eliminates the discriminatory equilibrium of the baseline model without impacting on any impartial equilibrium.  What remains to be verified is that it does not generate any new equilibrium. The latter, if exists, must induce different levels of effort from the agents. Without loss, let the effort profile $\bm\mu$ be $(\bar\mu, \ul\mu)$, and formalize the principal's problem under $\bm\mu$ (hereinafter, the primal problem), as: 
\[
\max_{\pi, a(\cdot)}\mathbb{E} \left[\tilde{a} \Delta \tilde{\theta} \mid  \bm{\mu},\pi,a(\cdot) \right]+\mu_w - \la I(\pi \mid \bm{\mu}) \text{ s.t. } \underbrace{\frac{1}{2} \geq \mathbb{E}\left[\tilde{a} \mid \bm\mu, \pi, a(\cdot)\right]}_\text{(\ref{eqn_quota})}.\vspace{-10pt}
\]
Note that in the objective function, only the term $I(\pi \mid \bm\mu)$ is convex in $\pi$ \citep{coverthomas},  while all remaining terms are linear in $(\pi, a(\cdot))$.  Moreover,  there clearly exists a $(\pi, a(\cdot))$ that strictly satisfies (\ref{eqn_quota}),  hence Slater's condition is met. As a result, strong duality holds,  and the primal problem can be solved using the Lagrangian method.  Let $\nu \geq 0$ denote the Lagrange multiplier associated with (\ref{eqn_quota}), and define the Lagrangian function as:
\[\mathcal{L}(\pi, a(\cdot), \nu)= \mathbb{E} \left[\tilde{a} (\Delta \tilde{\theta}-\nu) \mid  \bm{\mu},\pi,a(\cdot) \right]-\lambda I(\pi \mid \bm\mu) + \mu_w+\frac{\nu}{2}. \]
Write the primal problem as $\sup_{\pi, a(\cdot)} \inf_{\nu \geq 0} \mathcal{L}(\pi, a(\cdot), \nu)$, and the dual problem as $\inf_{\nu \geq 0} \sup_{\pi, a(\cdot)} \mathcal{L}(\pi, a(\cdot), \nu)$. 
Strong duality stipulates that these problems must have the same solution(s).  

Let $(\pi^*, a^*(\cdot), \nu^*)$ denote a solution, which clearly exists. A careful inspection of the problem $\sup_{\pi, a(\cdot)} \mathcal{L}(\pi, a(\cdot), \nu^*)$ reveals its equivalence to problem (\ref{eqn_aa}) at $\nu=\nu^*$.
In Lemma \ref{lem_aa}, we already characterized the solution to the last problem, showing, in particular, that the signal structure is of form $\pi^*: \{-1,0,1\}\rightarrow [0,1]$, and that it satisfies $X>Y$ if $\nu^*>0$ and $\bar\pi^*=1/2$.  To verify the last condition, notice that (\ref{eqn_quota}) must bind at the optimum, and so $(\pi^*, a^*(\cdot), \nu^*)$ must satisfy complementary slackness.  But then $\pi^*$ cannot simultaneously satisfy both agents' incentive compatibility constraints at $\bm\mu=(\bar\mu, \ul\mu)$, as argued in the main text.  This completes the proof that the use of quota does not generate new equilibria. \qed

\cleardoublepage

    \vspace*{16em}
    \begin{center}
        \Huge{
        Online Appendix for \\ ``Rationally Inattentive Statistical Discrimination: Arrow Meets Phelps''\\ by Federico Echenique and Anqi Li}
    \bigbreak
    \end{center}

\thispagestyle{empty}
\cleardoublepage

\appendix
\setcounter{section}{0}




\gdef\thesection{O.\arabic{section} }
\newtheorem{defnO}{Definition}
\renewcommand{\thedefnO}{O.\arabic{defnO}}
\newtheorem{exmO}{Example}
\renewcommand{\theexmO}{O.\arabic{exmO}}
\newtheorem{propO}{Proposition}
\renewcommand{\thepropO}{O.\arabic{propO}}


\section{Multiple tasks and occupational discrimination}\label{sec_multiple} 
This appendix extends the baseline model to encompass multiple tasks. The main takeaway from our analysis is that the ideas developed in the baseline model can be adapted to explain the rise and persistence of occupational discrimination. 
\vspace{-10pt}
\paragraph{Setup.} There are two tasks that need to be performed: $t=1,2$, each arriving randomly with  probability $\alpha^t \in (0, 1/2]$. The two tasks never arrive simultaneously,  thus it is always the case that exactly one of the tasks needs to be performed.

Agents can undertake multidimensional investments to improve their task-specific skills.  Agent $i$'s investment in skill $t$ is $\mu_i^t \in \{\ul\mu, \bar\mu\}$. Investment yields a high skill, $\theta_i^t=1$, with probability $\mu_i^t$, and a low skill, $\theta_i^t=0$, with the complementary probability $1-\mu_i^t$. Investing incurs a cost $C^t(\mu_i^t)$ to the agent, where $C^t(\ul\mu)=0$ and $C^t(\bar\mu)=C^t>0$. If the task that has to be performed is $t$, and agent $i$ is chosen to perform it, then that agent earns a reward  $\beta^t>0$, and the principal (who values the skill of the agent who is assigned to perform the task) gets a payoff of $\theta_i^t$.

The principal does not directly observe $\theta_i^t$s, but can acquire costly information about them. The signal he uses to screen agents for task $t$ is $\pi^t: \{-1,0,1\} \rightarrow [0,1]$. For each differential productivity value $\Delta\theta^t\coloneqq\theta_m^t-\theta_w^t$ between $m$ and $w$, the signal specifies the probability $\pi^t(\Delta\theta^t)$ that $m$ is assigned to perform task $t$.  

The game begins with all players moving simultaneously: the principal specifies the signal structures $\pi^t$, $t=1,2$, and agents decide whether to invest in each skill. After that, the task that needs to be performed arrives, and agents are screened according to the pre-specified signal structure.  We examine the pure strategy Bayes Nash equilibria of this game.  
\vspace{-10pt}

\paragraph{Preliminaries.}  First, it is useful to develop some notational conventions. For each $t\in\{1,2\}$, define $c^t \coloneqq C^t/(\alpha^t\beta^t\Delta\mu)$, and assume w.l.o.g.\ that $c^1 \leq c^2$. Intuitively, $c^t$ captures the effective cost that agents must incur in order to win the assignment of task $t$; $c^1 \leq c^2$ implies that skill 1 is more valuable than skill 2. 

In the main body of the paper, we defined three cutpoints in the attention cost parameter: $\lambda^*$, $\bar\lambda$, and $\ul\lambda$. As we increase $c$ --- the effective cost of exerting high effort --- these cutpoints must decrease, as more information  is needed to motivate agents to work hard. In what follows, we shall write the cutpoints as $\lambda^*(c)$, $\bar\lambda(c)$, and $\ul\lambda(c)$ in order to signify their dependence on $c$. The assumption $c^1 \leq c^2$ implies that the cutpoints are weakly higher for task 1 than for task 2. 

Next is our notion of specialization.

\begin{defnO}
Call an equilibrium \emph{non-specialized} if both agents adopt the same investment strategy. Call an equilibrium \emph{specialized} if one agent invests in skill 1 and the other agent invests in skill 2. 
\end{defnO}

One may think of a non-specialized equilibrium as the multidimensional analog of an impartial equilibrium, in which agents invest in the same skill and are screened indiscriminately by the principal. In a specialized equilibrium, however, agents invest in different skills and are screened differently. In the case where $m$ invests in skill 1 and $w$ in skill 2 (which will be our focus), the principal labels task 1 as ``traditionally male'' and task 2 as ``traditionally female,'' and screens $m$ and $w$ favorably for their respective tasks. Anticipating the discriminatory behavior on the part of the principal, agents invest in the skills that they are screened favorably for, which in turn reinforces the use of specialized screening.  In equilibrium, occupational segregation and stereotypes emerge, whereby $m$ and $w$ are believed to possess the needed skills for succeeding in different tasks,  and they do so indeed in spite of being identical ex ante.
\vspace{-10pt}
\paragraph{Results.} Our first result establishes the existence and uniqueness of specialized and non-specialized equilibria, analogous to Theorem 1 of the main text. 

\begin{propO}\label{prop_multi1}
Suppose that the regularity conditions stated in Theorem 1 of the main text hold for each $t\in\{1,2\}$, and hence that $0<\ul\lambda(c^t)<\lambda^*(c^t)<\bar\lambda(c^t)$ for each $t \in \{1,2\}$.  The following statements are true.
\begin{enumerate}[(i)]
\item A non-specialized equilibrium always exists. Generically, there is a unique non-specialized equilibrium, which  induces both agents to invest in both skills when $\lambda <\lambda^*(c^2)$, no agent to invest in any skill when $\lambda>\lambda^*(c^1)$, and both agents to invest in skill 1 but not skill 2 when $\lambda \in (\lambda^*(c^2), \lambda^*(c^1))$. 

\item A specialized equilibrium exists if and only if 
\[\frac{c^1}{c^2} \geq \frac{\bar\mu(1-\bar\mu)}{\ul\mu(1-\ul\mu)} \text{ and } \lambda \in [\ul\lambda(c^1), \bar\lambda(c^2)]. \]
Whenever a specialized equilibrium exists,  there is a unique specialized equilibrium in which $m$ invests in skill 1 and $w$ in skill 2.
\end{enumerate}
\end{propO}

In the non-specialized case, the signal structures used to screen agents become less Blackwell informative as the attention cost parameter increases.  When the attention cost parameter is below $\lambda^*(c^2)$,  screening is meticulous for both tasks,  and agents best-respond by investing in both skills. When the attention cost parameter is above $\lambda^*(c^1)$, screening is too noisy to incentivize investments in any skill. For the in-between case $\lambda \in (\lambda^*(c^2), \lambda^*(c^1))$, screening provides agents with just enough incentives to invest in the most valuable skill, but not enough incentives to invest in the other skill. 

The specialized case arises when the attention cost parameter is intermediate. To induce one and only one agent to invest in skill $t \in \{1,2\}$, we need $\lambda \in [\ul\lambda(c^t),\bar\lambda(c^t)]$. Taking intersections between skills and simplifying using $\bar\lambda(c^2)\leq \bar\lambda(c^1)$ and $\ul\lambda(c^2)\leq \ul\lambda(c^1)$, we obtain $[\ul\lambda(c^1), \bar\lambda(c^2)]$ as the parameter region that sustains specialization in an equilibrium. To ensure that $\ul\lambda(c^1) \leq \bar\lambda(c^2)$, the two tasks must be sufficiently similar in terms of their costs and benefits to the agents, i.e., $c^1/c^2 \geq \bar\mu(1-\bar\mu)/\ul\mu(1-\ul\mu)$. If the last condition fails, then both agents prefer to invest in the more valuable skill, hence the force behind specialization will unravel.

The second result concerns which of the specialized and non-specialized equilibria is the most profitable to the principal. The comparison is the most insightful when the two tasks are equally profitable to the principal, i.e., $\alpha^1 =\alpha^2$.

\begin{propO}\label{prop_multi2}
Let everything be as in Proposition \ref{prop_multi1}, and suppose that $\alpha^1=\alpha^2$. The following statements are true.
\begin{enumerate}[(i)]
\item When the game has a specialized equilibrium and a non-specialized equilibrium in which both agents invest in both skills, i.e., $\lambda \in [\ul\lambda(c^1), \bar\lambda(c^2)] \cap [0, \lambda^*(c^2)]$, the non-specialized equilibrium is the most profitable. 
\item When the game has a specialized equilibrium and a non-specialized equilibrium in which no agent invests in any skill, i.e., $\lambda \in [\ul\lambda(c^1), \bar\lambda(c^2)]\cap (\lambda^*(c^1), +\infty)$, the specialized equilibrium is the most profitable.
\item When the game has a specialized equilibrium and a non-specialized equilibrium in which both agents invest in skill 1 but not skill 2, i.e., $\lambda \in [\ul\lambda(c^1), \bar\lambda(c^2)]\cap (\lambda^*(c^2), \lambda^*(c^1)]$, the specialized equilibrium is the most profitable.
\end{enumerate}
\end{propO}

Parts (i) and (ii) of Proposition \ref{prop_multi2} are immediate from Theorem 2 of the main text. Part (iii) of this proposition is new. When the attention cost parameter is intermediate, each agent has just enough  incentives to invest in one skill, but no more. Now, who should invest in which skill? In the non-specialized case, both agents invest in the same skill. As a result, the principal has to compare and contrast them carefully every time a task needs to be assigned, which incurs a significant attention cost.  In the specialized case, the principal gives  stereotypical performance evaluations that favor $m$ in the assignment of the traditionally male task,  and $w$ in the assignment of the traditionally female task. Doing so saves on attention cost, while generating more revenue by inducing $m$ and $w$ to invest in their respective skills. 

\vspace{-10pt}
\paragraph{Implications.} There is ample evidence that men and women work on very different jobs, even within narrowly defined firms or industries \citep{ blau2017gender}. 
Recent sociological and experimental research such as \cite{correll2020inside} stresses the role of gender-stereotypical performance evaluations in sustaining and perpetuating this pattern, through coding and analyzing managers’ written reviews of employees at a Fortune 500 tech company. Stereotypical performance evaluation is also cited as a culprit for women's underrepresentation in STEM fields. On subjects such as math and sciences, a gender gap exists and is positively related to (female)  teacher’s bias in favor of boys \citep{lavy2018origins, moss2012science}.

Our analysis throw new light on these empirical findings by telling a story of endogenous stereotype formation and occupational segregation based on limited attention only, raising the possibility of curtailing these phenomena through modulating the availability of attentional resources.

\section{Alternative attention cost functions}\label{sec_cost}
We used mutual information to measure the cost of information acquisition in the main body of the paper.  The justification for this assumption differs, depending on whether one models  information acquisition as processing information or producing information.  

In the case of information processing,  it has been recognized by various authors that mutual information captures an ideal situation in which the decision maker already knows how to optimally encode states before processing the available information.  The result of optimal encoding is a property called ``compression invariance,'' whereby all payoff-equivalent states are treated as identical.  Reality, however, is full of situations in which payoff-equivalent states are treated differently based on their perceptual properties \citep{dean2023experimental}. To address this ``perceptual distance critique,'' \cite{caplin2022rationally} invent the class of uniform posterior separable (UPS) costs that nests mutual information as a special case.  While we cannot solve our model analytically for alternative UPS costs,  we have conducted numerical analysis and obtained qualitatively similar results to those under the mutual information cost. In Appendix \ref{sec_figure}, we depict the equilibrium regimes obtained under total information --- a UPS cost that is proposed by \cite{bloedel2020cost} and enjoys several desirable properties. 

To model information production,  several authors have advocated the use of prior invariant costs,  while \cite{bloedel2020cost}  provide a foundation for UPS costs based on sequential learning-proofness.  We take no stand on this debate, but only stress that prior dependence seems to be key to our results, as suggested by two exercises.

Consider first the following entropy-based cost function:  
\[K(\pi)\coloneqq I(\pi \mid q),\] 
which was proposed by \cite{denti2022experimental} to serve as an alternative to the mutual information cost. In words, $K(\pi)$ is the mutual information of the productivity state generated by a fixed ``reference'' prior $q \in \Delta (\{-1,0,1\})$ and the promotion recommendations prescribed by signal structure $\pi$.  By construction, $K$ is independent of the true prior distribution of the state generated by the agents' effort choices.  Meanwhile,  it becomes the mutual information cost when the reference prior equals the true prior,  thus serving as an ideal candidate for delineating the role of prior dependence in shaping our results. 

We say that the reference prior $q$ is \emph{symmetric between agents} if $q(1)=q(-1)$. Under a symmetric reference prior, the probability that $m$ is strictly more productive than $w$ equals the probability that $w$ is strictly more productive than $m$. 

The next proposition demonstrates that our game has only impartial equilibria and no discriminatory equilibrium when the reference prior is symmetric between agents. This result highlights the role of prior dependence in generating discriminatory equilibria.

\begin{propO}\label{prop_priorinvariant}
The aforementioned game has no discriminatory equilibrium if the reference prior is symmetric between agents. 
\end{propO}

The proof presented in Appendix \ref{sec_onlineproof} also delineates the region over which the game has an impartial, high-effort equilibrium and where it has an impartial, low-effort equilibrium. Notably, these regions are disjoint under Assumption 1(i) of the main text. 

Consider next a more stylized monitoring technology that fully reveals an agent's productivity value to the principal at a fixed cost $\kappa>0$. Note that under this technology, it is never optimal for the principal to monitor both agents. Instead, the principal should (randomly) monitor one agent. If the monitored agent has a high productivity value, then he or she should be promoted; otherwise the remaining agent should be be promoted. If this strategy is too costly, the principal should simply promote the agent who works harder a priori.  Anticipating the principal's monitoring decisions, agents make their effort choices. An equilibrium is \emph{impartial} if both agents make the same (random) effort choice and are monitored with equal probability. It is \emph{discriminatory} otherwise. 

The next proposition shows that while a discriminatory equilibrium can emerge over certain parameter regions under this alternative monitoring technology, it is never the most profitable equilibrium to the principal.  

\begin{propO}\label{prop_fixedcost}
Fix any values of $C$, $\bar\mu$ and $\ul\mu$ that satisfy Assumption 1 of the main text. Then the most profitable equilibrium to the principal is impartial for all $\kappa>0$. 
\end{propO}

The proof presented in Appendix \ref{sec_onlineproof} fully characterizes the equilibrium regimes.  

\section{Heterogeneous agents}\label{sec_heterogeneous}
In this appendix, we relax the assumption that agents are ex-ante identical and instead allow them to be heterogeneous. We consider three kinds of heterogeneity: heterogeneous effort costs, heterogeneous degrees of risk aversion, and group-specific monitoring cost. 

\vspace{-10pt}
\paragraph{Heterogeneous effort costs.} Let $C_i$ denote agent $i$'s cost of exerting high effort, $i \in \{m,w\}$, and assume w.l.o.g. that $C_m \leq C_w$. The case where $C_m=C_w$ was examined in the main body of the paper. 

To state our result properly, it is useful to recall a few concepts. In the main body of the paper, we defined three threshold values of the attention cost parameter: $\lambda^*$, $\ul\lambda$, and $\bar\lambda$, as decreasing functions of the effective effort cost. In the current context, define $c_i \coloneqq C_i/\Delta \mu$ as the effective effort cost that agent $i \in \{m,w\}$ incurs from exerting high effort, and note that $c_m \leq c_w$ by assumption.   

The next proposition characterizes the equilibria of our game when agents can differ in their (effective) effort costs.

\begin{propO}\label{prop_heterogeneous}
When $c_m \leq c_w$, our game has (i) an impartial equilibrium that sustains $(\ul \mu,\ul\mu)$ if $\lambda \geq \lambda^*(c_m)$; (ii) an impartial equilibrium that sustains $(\bar\mu,\bar\mu)$ if $\lambda \leq \lambda^*(c_w)$; (iii) a discriminatory equilibrium that  sustains $(\bar\mu, \ul\mu)$ if $\bar\mu+\ul\mu>1$ and 
$\lambda \in [\ul\lambda(c_w), \bar\lambda(c_m)]$, or if $\bar\mu+\ul\mu<1$,  $c_w/c_m>\bar\mu(1-\bar\mu)[\ul\mu(1-\ul\mu)]^{-1}$, and $\lambda \in [\ul\lambda(c_w), \bar\lambda(c_m)]$; (iv) a discriminatory equilibrium that sustains $(\ul\mu, \bar\mu)$ if $\bar\mu+\ul\mu>1$, $c_w/c_m<\ul\mu(1-\ul\mu)[\bar\mu(1-\bar\mu)]^{-1}$, and $\lambda \in [\ul\lambda(c_m), \bar\lambda(c_w)]$. 
\end{propO}

The messages conveyed by Proposition \ref{prop_heterogeneous} are largely to be expected.  When the effort cost differs between agents, the two regimes that sustain the high effort profile and low effort profile in an impartial equilibrium, respectively, may no longer be adjacent to each other.  This is because inducing both agents to work requires that we deter $w$ from shirking, i.e., $\lambda \leq \lambda^*(c_w)$, while inducing both of them to shirk requires that we discourage $m$ from working, i.e., $\lambda \geq \lambda^*(c_m)$.  Since $\lambda^*(\cdot)$ is a decreasing function, the two regimes are disjoint if $c_m<c_w$. 

As before, sustaining a discriminatory effort profile in an equilibrium is only possible when the attention cost parameter takes intermediate values. However, the exact conditions differ, depending on which agent is working and which one is shirking, resulting in a proliferation of cases.  Unlike the homogeneous case in which $\bar\mu+\ul\mu>1$ is always needed to sustain a discriminatory equilibrium, now inducing $m$ to work and $w$ to shirk becomes possible when $\bar\mu+\ul\mu<1$,  provided that the effort cost is significantly higher for $w$ than for $m$, i.e., $c_w/c_m>\bar\mu(1-\bar\mu)[\ul\mu(1-\ul\mu)]^{-1}$.  Inducing $w$ to work and $m$ to shirk becomes harder than before, in that in addition to $\bar\mu+\ul\mu>1$ and $\lambda \in [\ul\lambda(c_m), \bar\lambda(c_w)]$, we  need $c_w/c_m<\ul\mu(1-\ul\mu)[\bar\mu(1-\bar\mu)]^{-1}$ to hold.  The last condition stipulates that while it is more costly for $w$ to work than for $m$, the difference between their effort costs must not be excessive.

The proof of Proposition \ref{prop_heterogeneous} works by recognizing that in our model, heterogeneous effort costs operate only through adjusting the agents' incentive compatibility (IC) constraints.  In the meantime, they do not affect the principal's optimal signal structure for any given profile of effort choices,  and so do not alter the profitability ranking between impartial and discriminatory equilibria when the latter coexist.   To complete the equilibrium characterization,  all we need to do is to shift the blue and black line segments in Figure 2 of the main text,  to appropriately reflect the changes in effort costs.  The algebraic details are omitted but can be made available upon request.

\vspace{-10pt}
\paragraph{Heterogeneous degrees of risk aversion.} So far we have assumed that agents are risk neutral, in spite of the ample evidence suggesting that gender and ethnic minorities differ in their degrees of risk aversion from the majorities.  To capture this empirical regularity and examine its equilibrium consequences,  suppose that $m$ and $w$ are expected utility maximizers with Bernoulli utility functions $u_m$ and $u_w$, respectively.  For each $i \in \{m,w\}, $ define $\Delta u_i\coloneqq u_i(1)-u_i(0)$ as agent $i$'s utility gain from getting promoted. Then $m$ prefers to exert high effort rather than low effort if 
\[(1-\mu_w)X+\mu_wY \geq \frac{c_m}{\Delta u_m},\]
and $w$ prefers to exert high effort rather than low effort if 
\[\mu_mX+(1-\mu_m)Y \geq \frac{c_w}{\Delta u_w}.\]
Comparing the above IC constraints with those in the main text, we can see that
heterogeneous degrees of risk aversion operate in our model through the exact same channel as heterogeneous effort costs. Fortunately, we already know how to handle the latter by now.  

\vspace{-10pt}

\paragraph{Group-specific monitoring cost.} There are, of course, multiple ways to model group-specific monitoring cost. To capture this facet of reality within the paradigm of rational inattention, we allow the principal access to free, exogenous signals that differ across groups prior to information acquisition. \label{refgroupcost}

Specifically, suppose that before acquiring additional information, the principal observes a signal that informs him of  $m$'s productivity value but receives no signal about $w$. We model the signal as a distribution over the posterior beliefs that $m$ has a high productivity value. When both agents work, the posterior belief about $m$ --- a random variable --- must have a prior mean $\bar\mu$. For any given realization of the posterior, the principal's belief about $m$ differs that of $w$: the former equals the realized posterior while the latter remains at the prior $\bar\mu$.  From Lemma 1 of the main text (replace $\mu_m$ with the realized posterior), we know that the optimal signal acquired by the principal is discriminatory, despite both agents working equally hard.

To examine the impact of the exogenous signal on the agents' ex-ante incentives, we must solve the optimal signal posterior by posterior, and integrate the results across posteriors. This exercise is computational by nature and is better left for future work, using results of the current paper as building blocks.

\section{Commitment}\label{sec_commitment}
In the main body of the paper, we assumed that the principal moves simultaneously with the agents and cannot commit to the use of a signal structure. In this appendix, we examine an alternative game sequence whereby the principal moves first and commits to a signal structure.   Agents observe the signal structure chosen by the principal before making effort choices simultaneously among themselves. 

The next proposition shows that allowing the principal to commit makes it easier to sustain discrimination in equilibrium. 

\begin{propO}\label{prop_commitment}
Let everything be as in the main text, except that the game sequence has the principal first choosing a signal structure, as described above.  
\begin{enumerate}[(i)]
\item For any $\lambda \leq \lambda^*$, the equilibrium of the game induces the high effort profile using the same impartial signal structure as in the baseline model.  
\item For any $\lambda \in (\lambda^*, \bar\lambda]$,  the equilibrium of the game induces either the high effort profile using a discriminatory signal structure, or it induces the discriminatory effort profile using the same discriminatory signal structure as in the baseline model.
\item For any $\lambda>\bar\lambda$,  the equilibrium signal structure may be discriminatory, while that of the baseline model must be impartial.  
\end{enumerate}
\end{propO}

Recall that in the baseline model, the agents' IC constraints are generically slack, and the principal most prefers the impartial, high-effort equilibrium, followed by the discriminatory equilibrium, and then the impartial, low-effort equilibrium.  Together, these results imply that whenever the principal can induce the high effort profile without commitment, she will continue to do so with commitment, using the exact same signal structure as before. This is Part (i) of Proposition \ref{prop_commitment}.

Part (iii) of Proposition \ref{prop_commitment} is also easy to see.  Without commitment, the principal can only induce low effort when $\lambda>\bar\lambda$.  With commitment, she can still induce low effort using the same impartial signal structure as before, and she may be able to do better.  The signal structure in the second case can only be more discriminatory than that in the first case. 

Part (ii) of Proposition \ref{prop_commitment} is the most delicate.  Without commitment, the principal can induce both the discriminatory effort profile and the low effort profile when $\lambda \in (\lambda^*, \bar\lambda]$, and  she strictly prefers the first outcome to the second one.  With commitment, she faces a new possibility, that of inducing the high effort profile using a signal structure that makes one agent's IC constraint binding and the other agent's IC constraint slack.  In Appendix \ref{sec_onlineproof}, we show that the signal structure in the last case must be discriminatory,  as the Lagrange multiplier associated with the binding IC constraint distorts the optimal signal structure away from being impartial.  Consequently, even if the principal finds it optimal to induce the high effort profile, she will do so using a discriminatory signal structure rather than an impartial one. 

In practice, commitment to discriminatory practices is prohibited by law in many places. Whenever this is the case, the principal has a strong incentive to forgo the first-mover advantage and switch to the use of subjective monitoring.  One consequence of Proposition \ref{prop_commitment} is, then, related to curbing explicit discrimination. Arguably, a law that bans explicit discrimination may prevent the principal from using  the discriminatory screening device as in Proposition~\ref{prop_commitment}. It would, however, be ineffective against the sort of implicit discrimination we focused on in the main body of the paper.

\section{Additional robustness checks}\label{sec_additional}

\paragraph{Outside option.} So far we have restricted the principal to promoting either $m$ or $w$. Suppose now that the principal has a third option of promoting nobody (denoted by $o$), which generates a fixed payoff of $\theta_o \in (0,1)$ to him. \label{pageoutsideoption}

In the current context, a signal structure maps each profile of agents' productivity values to a probability distribution over the actions that the principal can take. For each productivity state $\bm\theta\coloneqq (\theta_m,\theta_w) \in \{0,1\}^2$, we denote the probability that the principal takes action $z \in \{m,w,o\}$ by $\pi_z(\bm\theta)$. A signal structure is \emph{impartial among the agents} if $\pi_m(\theta,\theta')=\pi_w(\theta', \theta)$ for all $\theta, \theta' \in \{0,1\}$. When this condition fails, the signal structure is \emph{discriminatory among the agents.} 

The next proposition examines how the presence of the outside option influence the optimal signal structure. 

\begin{propO}\label{prop_outsideoption}
Let everything be as above and fix any $\bm\mu \in \{\ul \mu, \bar\mu\}^2$. 
\begin{enumerate}[(i)]
\item As we increase $\theta_o$ slightly, the probability that the optimal signal structure for $\bm\mu$ recommends the outside option weakly increases in every state, and the increase is strict if each action is recommended with a strictly positive probability on average. 

\item The attention allocation among the agents is qualitatively similar as in the baseline model. Specifically, the optimal signal structure is impartial among the agents when the latter exert the same level of effort. When  $m$ works and $w$ shirks, suppose that $m$ is promoted with a strictly positive probability on average. Then the optimal signal structure must be discriminatory among the agents, promoting $w$ more often than $m$ only if she is strictly more productive.  
\end{enumerate}
\end{propO}

While the presence of the outside option does not qualitatively affect the attention allocation among agents, it makes the principal's average decision probabilities much harder to solve --- a well known challenge in the rational inattention literature. It also complicates the agents' IC constraints, which now depend on four probability differences: $\pi_m(1,\theta_{w})-\pi_m(0,\theta_w)$, $\theta_w \in \{0,1\}$ and $\pi_w(\theta_m, 1)-\pi_w(\theta_m, 0)$, $\theta_m \in \{0,1\}$, rather than two.\footnote{The IC constraints of $m$ and $w$ now become: \[\mu_w[\pi_m(1,1)-\pi_w(0,1)]+(1-\mu_w)[\pi_m(1,0)-\pi_m(0,0)] \geq c\] \[\text{ and }\mu_m[\pi_w(1,1)-\pi_w(1,0)]+ (1-\mu_m)[\pi_w(0,1)-\pi_w(0,0)] \geq c, \]\text{ respectively}. \label{fn_outsideoption}}  These challenges make it impractical to gain further insights, such as ranking the profits across different equilibria.

To make progress, we solve the model numerically. Our sense is that, with an outside option, the main takeaways of our paper can remain valid.  For example, when $\theta_o=.2$, $\ul\mu=.65$, $\bar\mu=.9$, $C=.1$, and the grid length is $.001$, the high-effort impartial equilibrium exists if $\lambda\leq .202$, and the low-effort impartial equilibrium exists otherwise. The discriminatory equilibrium exists if $[.173, .339]$, and it is the most profitable equilibrium if $\lambda \in [.202, .339]$. 

\vspace{-10pt}



\paragraph{Alternative solution concepts.} So far we have restricted agents to playing pure strategies (while imposing no restriction on the principal's strategy space).  It turns out that generically, allowing mixed strategies on the part of agents makes it easier to sustain discrimination in equilibrium, in the following sense.

\begin{propO}\label{prop_mse}
For all $\lambda \neq \lambda^*$, any mixed strategy equilibrium  with nontrivial randomization must feature  exactly one agent strictly mixing between high and low effort, while the other agent makes a deterministic effort choice. Consequently, the equilibrium must be discriminatory. 
\end{propO}

While the above proposition indicates that any mixed strategy equilibrium of our model must be discriminatory, it does not address when such an equilibrium exists, or whether it is more profitable than the other equilibria of the model. These questions are challenging to tackle analytically for the following reason: now that agents can randomize between high effort and low effort, their mixing probabilities --- which, in equilibrium, depend on the principal's choice of the signal structure --- will influence the probability distribution of the  productivity states.
Since the optimal signal structure depends on the attention cost parameter (recall Lemma 1 of the main text), agents' mixing probabilities in equilibrium and the resulting distribution of the productivity states also depend on this parameter. This added layer of dependence significantly complicates the revenue and cost functions of the principal, making it difficult to predict their behavior as we vary the attention cost parameter. As shown in Lemma 5 of the main text, this step is crucial for comparing the profits generated by different equiliabria. It is relatively straightforward in the baseline model where the distribution of the productivity states is independent of the attention cost parameter, once the agents' pure effort choices are given.\footnote{As an illustration, consider a candidate equilibrium in which $m$ works for sure and $w$ works with probability $\sigma$. Under this mixed effort profile, the probabilities that the productivity state $\Delta \theta$ equals $1$ and $-1$ are given by $A =\bar\mu(1-\nu)$ and $B=\nu (1-\bar\mu)$, respectively, where $\nu \coloneqq \ul\mu+\sigma \Delta \mu$. The effort profile can be sustained in an equilibrium if the optimal signal structure --- given by $X=f(\gamma)$ and $Y=AX/B$ --- satisfies $w$'s indifference condition $\bar \mu X+(1-\bar\mu) Y=c$. Since the solution to this system of equations depends on $\gamma \coloneqq \exp(1/\lambda)$, $A$ and $B$ depend on $\gamma$ as well. Substituting $A$ and $B$ into Lemma 5 of the main text, we can see that the principal's revenue and cost functions now exhibit an additional layer of dependence on $\gamma$ through equilibrium $\nu$. Ranking the profits generated by this mixed strategy equilibrium and other equilibria of the model is analytically challenging.  \label{fn_mse}}

The same issue arises when we allow agents to play correlated strategies. While it is easy to solve for the impartial equilibria of the game under this solution concept,\footnote{Correlating the agents' effort choices has no impact on the set of impartial equilibria generically. To verify, let $a$ denote the probability that both agents are recommended to work, $2b$ the probability that exactly one agent is recommended to work, and $1-a-2b$ the probability that both agents are recommended to shirk. When these marginal probabilities are all positive, the obedience constraint associated with the recommendation to work is $\frac{a}{a+b}[(1-\bar\mu)X+\bar\mu Y] + \frac{b}{a+b}[(1-\ul\mu)X+\ul\mu Y] \geq c$, and the one associated with the recommendation to shirk is $\frac{b}{1-(a+b)}[(1-\bar\mu)X+\bar\mu Y]+\frac{1-a-2b}{1-a-b}[(1-\ul\mu)X+\ul\mu Y] \leq c$. These are weighted averages of the incentive compatibility constraints in the baseline model. Simplifying using the symmetry of the signal, i.e., $X=Y=g(\gamma)$, yields  $g(\gamma)=c$, or equivalently $\lambda=\lambda^*$. \label{fn_correlatedequilibrium}} it is much harder to solve for the discriminatory equilibria and to compare their profitability with the impartial equilibria, for reasons discussed above.  \label{refcorrms}

\section{Proofs}\label{sec_onlineproof}

\paragraph{Proof of Proposition \ref{prop_multi1}.} First, notice that for each task $t \in \{1,2\}$ and effort profile $\bm\mu^t \coloneqq (\mu_m^t,\mu_w^t)$, the principal's problem is the same as in the baseline model.  What is left is to verify that the joint signal structure $(\pi^1, \pi^2)$ satisfies the  agents' IC constraints. Compared to the baseline model, agents can now commit two-step deviations that revise their effort choices for both tasks, in addition to one-step deviations that revise their effort choices for a single task. However, since the problems they face are additive separable across tasks, it suffices to deter one-step deviations only. Given this, we can treat the multidimensional problem as two separate single-dimensional problems --- an approach we will follow in the remainder of the proof. 

\vspace{5pt}

\noindent Part (i): The optimal signal structure for $((\bar\mu, \bar\mu), (\bar\mu, \bar\mu))$ is incentive compatible if and only if $\lambda \leq \min\{\lambda^*(c^1), \lambda^*(c^2)\}$. Since $c^1 \leq c^2$ and $\lambda^*(\cdot)$ is decreasing in its argument, the last condition is equivalent to $\lambda \leq \lambda^*(c^2)$. Likewise, the optimal signal structure for $((\ul \mu, \ul\mu),(\ul\mu, \ul\mu))$ is incentive compatible if and only if $\lambda \geq \max\{\lambda^*(c^1), \lambda^*(c^2)\}=\lambda^*(c^1)$, and the optimal signal structure for $((\bar\mu, \bar\mu), (\ul\mu,
\ul\mu))$ is incentive compatible if and only if $\lambda \in [\lambda^*(c^2), \lambda^*(c^1)]$. The optimal signal structure for $((\ul\mu, \ul\mu), (\bar\mu, \bar\mu))$ isn't incentive compatible unless $c^1=c^2$. 

\vspace{5pt}

\noindent Part (ii): The optimal signal structure for $((\bar\mu, \ul\mu),(\ul\mu, \bar\mu))$ is incentive compatible if and only if $\lambda \in \displaystyle \cap_{t=1}^2[\ul\lambda(c^t), \bar\lambda(c^t)]$. Since $\ul\lambda(\cdot)$ and $\bar\lambda(\cdot)$ are decreasing in their arguments,  $\cap_{t=1}^2[\ul\lambda(c^t), \bar\lambda(c^t)]\neq \emptyset$ if and only if $\ul\lambda(c^1) \leq \bar\lambda(c^2)$.  When proving Theorem 1 of the main text, we established that $\lambda \geq \ul\lambda(c)$ if and only if 
\[X \leq \bar{X}(c)=\frac{c\ul\mu}{\bar\mu},\]
and that $\lambda \geq \bar\lambda(c)$ if and only if 
\[X \geq \ul{X}(c)=\frac{c(1-\bar\mu)}{1-\ul\mu}. \]
Thus $\ul\lambda(c^1) \leq \bar\lambda(c^2)$ if and only if $\bar{X}(c^1) \geq \ul{X}(c^2)$, which, after simplifying, becomes:
\[\frac{c^1}{c^2} \geq \frac{\bar\mu(1-\bar\mu)}{\ul\mu(1-\ul\mu)}.   \] 

As a concluding remark, note that the method developed above also applies to situations in which agents undertake the same level of investment in one task but different levels of investment in the other task. For example, the optimal signal structure for $((\bar\mu, \bar\mu), (\bar\mu, \ul\mu))$ is incentive compatible if and only if $\lambda \leq \lambda^*(c^1)$ and $\lambda \in [\ul\lambda(c^2), \bar\lambda(c^2)]$. To save space, we choose not to exhaust all possibilities  and focus intead on specialized and non-specialized equilibria only. \qed
\vspace{-10pt}

\paragraph{Proof of Proposition \ref{prop_multi2}.} 
Parts (i) and (ii) of this proposition are immediate from Theorem 2 of the main text. To show Part (iii), we follow the notational convention developed in Lemma 5 of the main text. Specifically, we use $V(\bm{\mu};\gamma)$ and $I(\bm{\mu};\gamma)$ to denote the expected revenue and the mutual information cost generated by the optimal signal structure for $\bm{\mu}$, respectively. Write $\Delta V^1(\gamma)$ for $V((\bar\mu, \bar\mu); \gamma)-V((\bar\mu, \ul\mu); \gamma)$, $\Delta V^2(\gamma)$ for $V((\bar\mu, \ul\mu); \gamma)-V((\ul\mu, \ul\mu); \gamma)$, $\Delta I^1(\gamma)$ for $I((\bar\mu, \bar\mu); \gamma)-I((\bar\mu, \ul\mu);\gamma)$, and $\Delta I^2(\gamma)$ for $I((\bar\mu, \ul\mu);\gamma)-I((\ul\mu, \ul\mu);\gamma)$. 

We wish to show that 
\[\Delta V^1(\gamma)-\frac{1}{\ln\gamma} \Delta I^1(\gamma)-[\Delta V^2(\gamma)-\frac{1}{\ln\gamma}\Delta I^2(\gamma)]<0\quad \forall \gamma \in [\breve\gamma, +\infty).\]
Below we prove a stronger claim, namely $\Delta V^1(\gamma)<\Delta V^2(\gamma)$ and $\Delta I^1(\gamma)>\Delta V^2(\gamma)$ $\forall \gamma \in [\breve\gamma, +\infty)$. 

To show that $\Delta V^1(\gamma)<\Delta V^2(\gamma)$, recall from Lemma 5 of the main text that 
\[\Delta V^1(\gamma)=\frac{\Delta\mu}{\gamma+1}[\gamma-(\gamma-1)\bar\mu] \text{ and } \Delta V^2(\gamma)=\frac{\Delta\mu}{\gamma+1}[\gamma-(\gamma-1)\ul\mu].\]
Thus, 
\[\Delta V^1(\gamma)-\Delta V^2(\gamma)=-\frac{(\gamma-1)(\Delta\mu)^2}{\gamma+1}<0\]
as desired. 

To show that $\Delta I^1(\gamma)>\Delta I^2(\gamma)$ $\forall \gamma \in [\breve{\gamma}, +\infty)$, notice that the claim is clearly true at $\gamma=\breve\gamma$, since $\Delta I^1(\breve\gamma)>0$ and $\Delta I^2(\breve\gamma)<0$. It is also true when $\gamma$ is very large, since 
\begin{align*}
&\lim_{\gamma \rightarrow +\infty} \Delta I^1(\gamma)-\Delta I^2(\gamma)\\
&= 2 [\bar\mu(1-\bar\mu)+\ul\mu(1-\ul\mu)]\ln 2-2[A\ln\left(\frac{A+B}{A}\right)+B\ln\left(\frac{A+B}{B}\right)]\\
\tag{Verify using Mathematica}&>0.
\end{align*}
Then from 
\begin{align*}
&\frac{d}{d\gamma} \Delta I^1(\gamma)-\Delta I^2(\gamma)\\
\tag{$\because$ Lemma 5}&=\frac{\Delta \mu (1-2\bar\mu)\ln \gamma}{(\gamma+1)^2}-\frac{\Delta \mu (1-2\ul\mu)\ln \gamma}{(\gamma+1)^2}\\
&=-\frac{2(\Delta\mu)^2\ln \gamma}{(\gamma+1)^2}<0,
\end{align*}
it follows that $\Delta I^1(\gamma)-\Delta I^2(\gamma)$ is everywhere positive on $[\breve\gamma, +\infty)$ as desired.  \qed

\vspace{-10pt}
\paragraph{Proof of Proposition \ref{prop_priorinvariant}.} We first solve for the optimal signal structure obtained under $K$. To this end, we fix a profile $\bm\mu \in \{\ul\mu, \bar\mu\}^2$ of effort choices and use $p$ to denote the true prior distribution of the productivity state under $\bm\mu$. Addtionally, let $\bar \pi_q$ denote the average probability that a signal structure $\pi$ recommends $m$ for promotion under the reference prior $q$.  As shown by \cite{denti2022experimental}, the optimal signal structure for $\bm{\mu}$ is fully pinned down by (i) an augmented version of the multinomial logit formula:
\[\pi(\Delta\theta)=\frac{\bar\pi_q\exp\left(\frac{\Delta\theta}{\lambda}\frac{p(\Delta\theta)}{q(\Delta\theta)}\right)}{\bar\pi_q\exp\left(\frac{\Delta\theta}{\lambda}\frac{p(\Delta\theta)}{q(\Delta\theta)}\right)+1-\bar\pi_q} \text{ }\forall \Delta \theta \in \{-1,0,1\},\]
and (ii) Bayes's plausibility under the reference prior: 
 \[ \sum_{\Delta\theta \in \{-1,0,1\}} q(\Delta\theta)\pi(\Delta\theta)=\bar\pi_q.\]
Solving these equations simultaneously yields: 
 \[\bar\pi_q=\frac{(\alpha-1)\beta q(1)-(\beta-1)q(-1)}{(\alpha-1)(\beta-1)[q(1)+q(-1)]},\]
where
 \[\alpha \coloneqq \exp\left(\frac{p(1)}{\lambda q(1)}\right) \text{ and } \beta \coloneqq \exp\left(\frac{p(-1)}{\lambda q(-1)}\right).\]

What matter for agents' incentives are:  $X \coloneqq \pi(1)-\pi(0)$ and $Y \coloneqq \pi(0)-\pi(-1)$, calculated as follows: 
\[X=\frac{\alpha \bar\pi_q}{\alpha \bar\pi_q+1-\bar\pi_q}-\bar\pi_q=\frac{(\alpha-1)\bar\pi_q(1-\bar\pi_q)}{\alpha \pi_q+1-\bar\pi_q}\] and 
\[Y=\bar\pi_q-\frac{\bar\pi_q}{\bar\pi_q+(1-\bar\pi_q)\beta}=\frac{(\beta-1)\bar\pi_q(1-\bar\pi_q)}{\pi_q+\beta(1-\bar\pi_q)}.\]
Dividing $X$ by $Y$ yields: 
\[\frac{X}{Y}=\frac{\alpha-1}{\beta-1}\frac{\bar\pi_q+\beta(1-\bar\pi_q)}{\alpha \pi_q+1-\bar\pi_q}=\frac{\alpha-1}{\beta-1}\frac{\frac{\beta-1+(\alpha-1)\beta}{2(\alpha-1)}}{\frac{(\alpha-1)\beta+\beta-1}{2(\beta-1)}}=1,\]
where the second equality exploits the expression for $\bar\pi_q$. Substituting this result into the agents' IC  constraints shows that either both agents prefer to work rather than shirk, or they both prefer to shirk than to work. The discrimininatory effort profile, on the other hand, cannot be sustained in any equilibrium. 

To fully characterize the equilibrium regime, note that $\alpha=\beta=\exp(p(1)/\lambda q(1))$ if the underlying effort profile is impartial. Simplifying the expressions for $X$ and $Y$ accordingly and substituting the result into the agents' IC constraints, we obtain that $(\bar\mu, \bar\mu)$ can arise in an equilibrium if $\lambda \leq \left(\ln\left(\frac{q(1)g^{-1}(c)}{\bar\mu(1-\bar\mu)}\right)\right)^{-1}$, and that $(\ul\mu, \ul\mu)$ can arise in an equilibrium if $\lambda \geq \left(\ln\left(\frac{q(1)g^{-1}(c)}{\ul\mu(1-\ul\mu)}\right)\right)^{-1}$. The two regimes are disjoint under Assumption 1(i) of the main text. \qed

\vspace{-10pt}

\paragraph{Proof of Proposition \ref{prop_fixedcost}.} For each $i \in \{m,w\}$, write $\pi_i$ for the probability that agent $i$ is monitored, and $\mu_i$ for the probability that agent $i$ has a high productivity value, given his or her (random) effort choice. In Appendix \ref{sec_cost}, we already demonstrated that it is never optimal for the principal to monitor both agents. Instead, the principal should (randomly) monitor one agent and then promote the monitored agent if his productivity is high, or promote the remaining agent if the monitored agent's productivity is low. The resulting gain from monitoring agent $i$ is $\mu_i+(1-\mu_i)\mu_{-i}$, which is independent of $i$. Consequently, the principal must be indifferent between monitoring either $m$ or $w$. Finally, monitoring is preferable to not monitoring if the above quantity is greater $\max_i \mu_i$ --- the gain from promoting the ex-ante more productive agent --- plus the cost of monitoring $\kappa$. 

Turning to agents' incentives, we note that agent $i$ can increase his probability of winning the promotion by   $\pi_i\Delta \mu$ if he chooses to work rather than shirk. Working is preferable to shirking if $\pi_i \Delta \mu \geq C$, or, equivalently, $\pi_i \geq c$. 
 
Intersecting players' best response functions provides a full characterization of  equilibrium regimes. To ensure a fair comparison with the baseline model, we assume that the production technology satisfies Assumption 1 of the main text. As noted earlier, this assumption implies that $c<1/2$ and $\bar\mu(1-\bar\mu)<\ul\mu(1-\ul\mu)$. 

Consider first impartial equilibria, which always exist and is generically unique. When $\kappa > \ul \mu (1-\ul\mu)$, the impartial equilibrium induces both agents to shirk; the principal monitors no agents and earns an expected profit of $\ul\mu$. When $\kappa \in (\bar\mu(1-\bar\mu), \ul\mu(1-\ul \mu))$, the equilibrium induces both agents to work with equal probability such that $\mu(1-\mu)=\kappa$; the principal is indifferent between monitoring and no monitoring, resulting in an expected profit of $\mu$; he monitors each agent with probability $c$, making them both indifferent between working and shirking. Finally, when $\kappa<\bar\mu (1-\bar\mu)$, the equilibrium induces both agents to work; the principal monitors each agent with probability $1/2$, and his expected profit equals $\bar\mu+(1-\bar\mu)\bar\mu-\kappa$.

Second, a discriminatory equilibrium in which $\mu_m>\mu_w$ exists (but is not necessarily unique) when $\kappa < \ul \mu (1-\ul\mu)$. When $\kappa \in (\ul\mu (1-\bar\mu), \ul\mu (1-\ul \mu))$, there exists a discriminatory equilibrium in which $\mu_w=\ul\mu$ and $\mu_m$ solves $(1-\mu_m)\ul \mu=\kappa$; the principal is indifferent between monitoring and no monitoring, resulting in an expected profit of $\mu_m$; he monitoring $m$ with probability $\pi_m=c$ and $w$ with probability $\pi_w<c$, which in turn makes $m$ indifferent between working and shirking while $w$ strictly prefers to shirk. When $\kappa \in (\ul\mu(1-\bar\mu), \bar\mu(1-\bar\mu))$, there exists another discriminatory equilibrium in which $\mu_m=\bar\mu$ and $\mu_w$ solves $(1-\bar\mu)\ul \mu_w=\kappa$; the principal is indifferent between monitoring and no monitoring, resulting in an expected profit of $\bar\mu$; he monitors $m$ with probability $\pi_m>c$ and $w$ with probability $\pi_w=c$, which in turn makes $m$ work and $w$ indifferent between working and shirking. Finally, when $\kappa<\ul\mu (1-\bar\mu)$, there exists a discriminatory equilibrium in which $m$ works, $w$ shirks, and the principal monitors only $m$; all players' incentives are strict, and the principal's expected profit equals $\bar\mu+(1-\bar\mu)\ul\mu-\kappa$.

By intersecting these equilibrium regimes, we can conclude that the principal's most preferred equilibrium is never discriminatory, thus completing the proof. \qed 

\vspace{-10pt}

\paragraph{Proof of Proposition \ref{prop_commitment}.} Parts (i) and (iii) of this proposition have already been established in Appendix \ref{sec_commitment}.  To show Part (ii),  notice that when $\lambda \in (\lambda^*,  \bar\lambda]$, the principal can induce $(\bar\mu, \ul \mu)$ and $(\ul\mu,\ul\mu)$ the same way as in the main text,  and she prefers the first outcome to the second one.  The only way to do better is to induce $(\bar\mu, \bar\mu)$ using a signal structure that makes one agent's IC constraint binding and the other agent's IC constraint slack.

Suppose w.l.o.g.  that it is $m$ whose IC constraint is binding and $w$ whose IC constraint is slack.  In that case, the principal's problem can be formalized as follows: 
\begin{align*}
\max_{\pi:\{-1,0,1\} \rightarrow [0,1]} &\sum_{\Delta\theta \in \{-1,0,1\}} p(\Delta \theta)\pi(\Delta\theta)\Delta\theta +\bar\mu- \lambda I(\pi \mid p) \\
\tag{IC$_m$}\text{s.t. }  & (1-\bar\mu) [\pi(1)-\pi(0)]+\bar\mu [\pi(0)-\pi(-1)] \geq c,
\end{align*}
where the term $p(\Delta\theta)$ in the objective function denotes the probability that $\Delta\theta$ occurs under $(\bar\mu,\bar\mu)$, and $I(\pi \mid p)$ denotes the mutual information cost when the underlying states follow distribution $p$.  Since the objective function is concave in $\pi$ and the constraint is linear in $\pi$,  strong duality holds. Let $\nu_m>0$ denote the Lagrange multiplier associated with the constraint, and rewrite the principal's problem as: 
\[\max_{\pi:\{-1,0,1\} \rightarrow [0,1]} \begin{Bmatrix}p(1)\pi(1)[1+\frac{\nu_m(1-\bar\mu)}{p(1)}] 
+p(0)\pi(0)[0+\frac{\nu_m (2\bar\mu-1)}{p(0)}]\\
+p(-1)\pi(-1)[-1-\frac{\nu_m \bar\mu}{p(-1)}]
\end{Bmatrix}-\lambda I(\pi \mid p).\] 
Since this is a binary RI decision problem with the principal's payoffs in the various states being: 
 \[v(1)=1+\frac{\nu_m(1-\bar\mu)}{p(1)},v(0)=\frac{\nu_m (2\bar\mu-1)}{p(0)}, \text{ and } v(-1)=-1-\frac{\nu_m \bar\mu}{p(-1)},\]
 the solution must satisfy the multinomial logit formula: 
\[\pi(\Delta \theta)=\frac{\bar\pi \exp\left(v(\Delta\theta)/\lambda\right)}{\bar\pi \exp\left(v(\Delta\theta)/\lambda\right)+1-\bar\pi} \text{ } \forall \Delta\theta \in \{-1,0,1\}.\]
In the case where the solution is impartial, we must have $\bar\pi=1/2$,  as well as 
 \[\pi(0)=\frac{\bar\pi \exp\left(v(0)/\lambda\right)}{\bar\pi \exp\left(v(0)/\lambda\right)+1-\bar\pi}>1/2,\]
where the inequality uses the fact that $\bar\mu>1/2$ and $\nu_m>0$.  Since impartiality requires that $\pi(0)=1/2$, we have reached a contradiction, hence the solution must be discriminatory, as desired. \qed

\vspace{-10pt}
\paragraph{Proof of Proposition \ref{prop_outsideoption}.} Part (i) of this proposition is immediate from Proposition 3 of \cite{matvejka2015rational}. 

To prove Part (ii), we fix any $\bm\mu \in\{\ul\mu, \bar\mu\}^2$ and write $\bar\pi_z$ for the average probability that the optimal signal structure for $\bm\mu$ recommends action $z \in \{m,w,o\}$. By the multinomial logit formua of \cite{matvejka2015rational}, we have: 
\[\pi_z(\bm\theta)=\frac{\bar\pi_{z} \exp(\theta_z/\lambda)}{\sum_{z' \in \{m,w,o\}} \bar\pi_{z'} \exp(\theta_{z'}/\lambda)} \text{ } \forall z \in \{m,w,o\} \text{ and } \bm\theta \in \{0,1\}^2. \]
When $\mu_m=\mu_w$, $\bar\pi_m=\bar\pi_w$ clearly holds. Substituting this into the above formula yields $\pi_m(\theta,\theta')=\pi_w(\theta',\theta)$ $\forall (\theta,\theta')$, hence the optimal signal is impartial among the agents. When $\mu_m>\mu_w$ and $\bar\pi_m>0$, $\bar\pi_m > \bar\pi_w$ clearly holds. If $\bar\pi_w=0$, then $\pi_w(\bm\theta)=0$ while  $\pi_m(\bm\theta)>0$ $\forall \bm\theta$. Consequently, the optimal signal structure is discriminatory among the agents. If, instead,  $\bar\pi_w>0$, then 
\[\frac{\pi_m(\bm\theta)}{\pi_w(\bm\theta)}=\exp\left(\frac{\theta_m-\theta_w}{\lambda}\right)\frac{\bar\pi_m}{\bar\pi_w} \text{ } \forall \bm\theta \in \{0,1\}^2. \]
The optimal signal structure must be discriminatory, as $\pi_m(\bm\theta)<\pi_w(\bm\theta)$ only if $\bm\theta=(0,1)$, and the inequality is flipped for any other $\bm\theta$.  \qed 

\vspace{-10pt}

\paragraph{Proof of Proposition \ref{prop_mse}.} When mixed strategies are allowed, let $\sigma_i \in [0,1]$ denote the probability that agent $i \in \{m,w\}$ exerts high effort, and assume w.l.o.g. that $\sigma_m \geq \sigma_w$.  Write $\nu(\sigma_i)$ for the probability $\ul\mu+\sigma_i \Delta \mu$  that  agent $i$ has a high productivity value. Then 
$A \coloneqq \nu(\sigma_m)[1-\nu(\sigma_w)]$ is the probability that the productivity state $\Delta \theta=1$  and $B \coloneqq \nu(\sigma_w)[1-\nu(\sigma_m)]$ is the probability that $\Delta\theta=-1$. 
Substituting $A$ and $B$ into Lemma 1 of the main text, we obtain the principal's optimal signal structure for any given profile the agents' (mixed) strategies, which must satisfy
\begin{equation}\label{eqn_signal}
X=\frac{(\gamma A-B)(\gamma B-A)}{(\gamma^2-1)(A+B)A} \text{ and } Y=\frac{A}{B}X. 
\end{equation}
if it is nondegenerate (if the signal structure is degenerate, then the equilibrium must discriminatory, and we are done). Meanwhile,  agents must be indifferent if they strictly mix between high effort and low effort. For $m$, this happens when
\begin{equation}\label{eqn_mindiff}
[1-\nu(\sigma_w)]X+\nu(\sigma_w)Y=c.
\end{equation}
For $w$, this happens when 
\begin{equation}\label{eqn_windiff}
\nu(\sigma_m)X+[1-\nu(\sigma_m)]Y=c.
\end{equation}

In case both agents strictly mix, solving (\ref{eqn_mindiff}) and (\ref{eqn_windiff}) simultaneously  yields either $X=Y=c$,  or $\nu(\sigma_m)+\nu(\sigma_w)=1$, or both.  In all scenarios, we must have $A=B$ and $X=Y$. Substituting this into (\ref{eqn_signal}) and solving yields $g(\gamma)=c$, or, equivalently, $\lambda=\lambda^*$. 

When $\lambda \neq \lambda^*$, (\ref{eqn_mindiff}) and (\ref{eqn_windiff}) cannot hold simultaneously. Consequently, any mixed strategy equilibrium with nontrivial randomization must have exactly one agent randomizing while the other agent making a deterministic effort choice. Such an equilibrium must be discriminatory. \qed

 \section{Figures}\label{sec_figure}
 
 \begin{figure}[!h]
\centering
     \includegraphics[width=.9\linewidth]{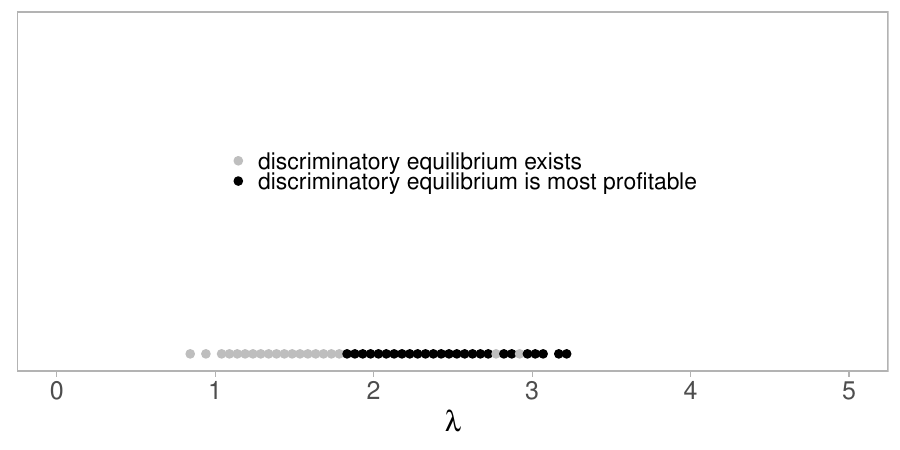}
    \caption{Equilibrium regime under the total information cost: $\underline{\mu}=0.55$, $\bar\mu=.65$, $C=.03$,  $\lambda$ ranges from $.1$ to $5$, and \# of grids is 100. An impartial equilibrium always exists; it sustains the high effort profile if $\mu<1.67$ and the low effort profile if $\mu>1.67$.  }\label{figure_totalinfo}
	\end{figure} 

\begin{spacing}{.8}
\bibliographystyle{aer} 
\bibliography{Discrimination.bib}
\end{spacing}

\end{document}